\documentclass[10pt,a4paper]{amsart}

\usepackage[latin1]{inputenc}
	\usepackage[toc,page]{appendix}
	\usepackage[numbers]{natbib}
	\usepackage{amsthm}
	\usepackage{enumerate} %
	\usepackage{amsfonts,mathrsfs}
	\usepackage{graphicx,epstopdf}
	\usepackage{float} %
	\usepackage{hyperref}
\usepackage{amsmath,amssymb}
\usepackage{xparse}

\makeatletter
\def\amsbb{\use@mathgroup \M@U \symAMSb}
\makeatother

\usepackage{bbold}

\theoremstyle{definition}
\newtheorem{defn}{Definition}[section]
\newtheorem{thrm}{Theorem}[section]
\newtheorem{cor}{Corollary}[section]
\newtheorem{prop}{Proposition}[section]
\newtheorem{lem}{Lemma}[section]
\newtheorem{ass}{Assumption}[section]
\newtheorem{ex}{Example}[section]
\newtheorem{remark}{Remark}[section]
\newtheorem{property}{Property}[section]

	\newcommand{\pp}{\amsbb{P}}
			\NewDocumentCommand\qq{o}
				{%
					\amsbb{Q}\IfValueT{#1}{_{#1}}
				}
	\newcommand{\rr}{\amsbb{R}}
	\newcommand{\nn}{\amsbb{N}}

	\newcommand{\aaaa}{\mathscr{A}}
	\newcommand{\cccc}{\mathscr{C}}	
	
	\newcommand{\mmmm}{\mathscr{M}}	
	\newcommand{\qqqq}{\mathscr{Q}}	
	\newcommand{\rrrr}{\mathscr{R}}	
		
	\newcommand{\xxxxx}{\mathscr{X}}
	\newcommand{\yyyyy}{\mathscr{Y}}

\NewDocumentCommand\NN{ooo}{
				{
				\mathcal{NN}^{\sigma}_{
					\IfNoValueT{#1}{\infty}\IfValueT{#1}{{#1}}
					,\IfNoValueT{#3}{d}\IfValueT{#3}{{#3}}
					,
				\IfNoValueT{#2}{d}\IfValueT{#2}{{#2}}
					}
				}
		}

\NewDocumentCommand\ES{oo}{\operatorname{ES}
							\IfValueT{#2}{_{#2}}\IfValueF{#2}{_{a}}
							\IfValueT{#1}{\left({#1}\right)}
							}
\NewDocumentCommand\ledot{}{_{\cdot}}

\NewDocumentCommand\utility{}{
	{{
		\mathcal{U}%
	}}
}
\NewDocumentCommand\riskutility{o}{
	\IfValueF{#1}{{\rho^{\utility}}}
	\IfValueT{#1}{{{\tilde{\rho}}^{\utility}}}
}
\NewDocumentCommand\px{o}{
					{{\pi}^{\utility,\lambda}_{\pp,\rho}\left(
						\IfValueF{#1}{{X}}\IfValueT{#1}{{#1}}
						\right)}
						}

\NewDocumentCommand{\dlprod}{oo}{
	\langle
	\IfValueT{#1}{#1}
	,
	\IfValueT{#2}{#2}
	\rangle
}

\NewDocumentCommand{\dom}{mo}{
	\IfValueT{#2}{\operatorname{int}\left(}
	\operatorname{dom}\left({#1}\right)
	\IfValueT{#2}{\left)}
}

		\NewDocumentCommand\epi{m}{
		\operatorname{epi}\left(
		{#1}
		\right)
	}

		\NewDocumentCommand\eemes{moo} %
		{%
			\amsbb{E}\IfValueT{#2}{_{#2}}\IfValueT{#3}{^{#3}}
			\left[#1
			\right]%
		}
		\NewDocumentCommand\ee{mooo} %
		{%
			\amsbb{E}\IfValueT{#3}{_{#3}}\IfValueT{#4}{^{#4}}
			\left[#1
				\IfValueT{#2}{\middle|{#2}}
			\right]%
		}
			\NewDocumentCommand\eee{moo} %
		{%
			\mathscr{E}\IfValueF{#3}{_{\qqqq}}\IfValueT{#3}{_{#3}}
			\left(#1
			\IfValueT{#2}{\middle|{#2}}
			\right)%
		}
	
		\NewDocumentCommand\eeRAp{m}{
			\rrrr\left(
			{#1}
			\right)
			}
			
		\NewDocumentCommand\eeRA{moooo}
		{%
			\IfValueT{#4}{{#4}}\IfValueF{#4}{\rrrr}
			\IfValueT{#5}{_{#5}}%
			\IfValueT{#3}{^{#3}}\IfValueF{#3}{^{\lambda}}
			\left(#1
			\IfValueT{#2}{\middle|{#2}}
			\right)%
		}
						\NewDocumentCommand\eeRArn{mooo}
				{%
					\IfValueT{#4}{{#4}}\IfValueF{#4}{\rho}_{\qq}
					\IfValueT{#3}{^{#3}}\IfValueF{#3}{^{\utility,\lambda}}
					\left(#1
					\IfValueT{#2}{\middle|{#2}}
					\right)%
				}
	
	\NewDocumentCommand\eenergy{mooo}
	{%
			\mathscr{E}\IfValueT{#3}{^{#3}}\IfValueF{#3}{^{\utility,\lambda}}
		\IfValueT{#4}{_{#4}}\IfValueF{#4}{_{\pp,\rho}}
		\left[#1
		\IfValueT{#2}{\middle|{#2;\Phi}}
		\right]%
	}

\NewDocumentCommand\PH{o}{
	{\operatorname{PH}
		\IfValueT{#1}{\left({#1}\right)}
	}
}

\NewDocumentCommand\SparseSpace{oooo}{
		{
		\mathcal{S}^{2,
		\IfValueT{#4}{{#4}}
		\IfValueF{#4}{{d}}
		}
		\IfValueT{#1}{_{\qq,#1}}
		\IfValueF{#1}{_{\qq,\epsilon}}
		\left(
		\IfValueT{#2}{{#2}}
		\IfValueF{#2}{{\mathscr{F}}}
			,
		\IfValueT{#3}{{\rrd[{#3}]}}
		\IfValueF{#3}{{\rrd[D]}}
		\right)
		}
	}
		\NewDocumentCommand\SparseSpacesigma{ooo}{
			{
				\mathcal{S}^{2,d}
				\IfValueT{#2}{_{\qq,#2}}
				\IfValueF{#2}{_{\qq,\epsilon}}
				\left(
				\IfValueT{#1}{{#1}}
				\IfValueF{#1}{{\mathscr{G}}}
				,
				\IfValueT{#3}{{\rrd[{#3}]}}
				\IfValueF{#3}{{\rrd[D]}}
				\right)
			}
		}

			\NewDocumentCommand\arginf{mo}
			{
				\underset{#1}{
			\IfValueT{#2}{\operatorname{argmin}}\IfValueF{#2}{\operatorname{arginf}}
			}	\,
			}
			\NewDocumentCommand\argsup{mo}
			{
				\underset{#1}{
					\IfValueT{#2}{\operatorname{argmax}}\IfValueF{#2}{\operatorname{argsup}}
				}	\,
			}
			\NewDocumentCommand\esssup{mo}
			{
				\underset{#1}{
					\IfValueT{#2}{\operatorname{essmax}}\IfValueF{#2}{\operatorname{esssup}}
				}	\,
			}

\NewDocumentCommand\Glim{oo}{\Gamma\mbox{-}\lim\limits_
	\IfValueT{#1}{{{#1} \uparrow {#2}}}\IfValueF{#1}{{n \uparrow \infty}}
}
\NewDocumentCommand\Klim{oo}{\operatorname{K}\mbox{-}\lim\limits_
	\IfValueT{#1}{{{#1} \uparrow {#2}}}\IfValueF{#1}{{n \uparrow \infty}}
}

			\NewDocumentCommand\mmm{oo}{
				\mathscr{M}^2_{0,loc}\left(
				\IfValueT{#1}{#1}\IfValueF{#1}{\pp};
				\IfValueT{#2}{#2}\IfValueF{#2}{\mathscr{F}_{\cdot}}
				\right)
			}
		
\NewDocumentCommand\Prox{m}{\operatorname{Prox}_{{#1}}}
			\NewDocumentCommand\aaarho{o}
			{
				\mathscr{A}
				\IfValueT{#1}{_{#1}}
				\IfValueF{#1}{_{\rho}}
			}
			\newcommand{\reflecetedrho}{\mathbin{\reflectbox{$\rho$}}}
			\NewDocumentCommand\cormes{ooo}
		{
			\IfValueF{#3}{{\rho}}
			\IfValueT{#3}{{\reflecetedrho}}
			\IfValueT{#1}{^{#1}}
			\IfValueF{#1}{^{\lambda}}
			\IfValueT{#2}{
			\left(
			{#2}
			\right)
			}
		}

				\NewDocumentCommand\ffff{oo}
					{%
						\mathscr{F}
						\IfValueT{#1}{_{#1}}\IfValueF{#1}{}%
						\IfValueT{#2}{^{#2}}%
					}
				\NewDocumentCommand\gggg{o}
					{%
						\mathscr{G}
						\IfValueT{#1}{_{#1}}
					}
				\NewDocumentCommand\hhhh{o}
					{%
						\mathscr{H}
						\IfValueT{#1}{_{#1}}
					}
				\NewDocumentCommand\hhh{o}
				{%
					\mathcal{H}
				}
		
	\newcommand{\rrbar}{{\overline{\rr}}}
		\NewDocumentCommand\rrext{}{
			{(-\infty,\infty]}
		}	
		\NewDocumentCommand\rrd{o}
		{%
			\rr\IfValueT{#1}{^{#1}}\IfValueF{#1}{^{d}}
		}
		\NewDocumentCommand\Lp{mooo}
				{%
					\operatorname{L}\IfValueF{#4}{_{\qq}}\IfValueT{#4}{_{#4}}%
					\IfValueT{#3}{^{#3}}\IfValueF{#3}{^{2}}
					\left(#1
						;\IfValueF{#2}{\hhh}\IfValueT{#2}{\rrd[{#2}]}
					\right)%
				}
						\NewDocumentCommand\Lpsub{mooo}
				{%
					\operatorname{L}\IfValueF{#4}{_{\pp}}\IfValueT{#4}{_{#4}}
					\IfValueT{#3}{^{#3}}\IfValueF{#3}{^{2}}
					\left(#1
					;\IfValueF{#2}{\hhh}\IfValueT{#2}{{#2}}
					\right)%
				}

		\NewDocumentCommand\QQ{o}
			{%
				\mathcal{Q}\IfValueT{#1}{_{#1}}\IfValueF{#1}{_{\rho}}
			}

		\NewDocumentCommand\rrdbar{o}
		{%
			\IfValueT{#1}{[}\IfValueF{#1}{(}-\infty,\infty]
		}

\title{\Large \bf  Partial Uncertainty and Applications to Risk-Averse Valuation} %
\thanks{Department of Mathematics, ETH Z\"{u}rich, HG G 32.3, R\"{a}mistrasse 101, 8092 Z\"{u}rich.  email: anastasis.kratsios@math.ethz.ch}

\date{%
	\today \\ %
}

\usepackage{fancyhdr}
\begin{document}
\author{Anastasis Kratsios
}
\title{\Large \bf Partial Uncertainty and Applications to Risk-Averse Valuation} %
\lhead{A. Kratsios
}
\chead{\small  Partial Uncertainty}
\rhead{\small \today}%

\date{%
	\today %
}

\maketitle

\begin{abstract}
This paper introduces an intermediary between conditional expectation and conditional sublinear expectation, called $\rrrr$-conditioning.  The $\rrrr$-conditioning of a random-vector in $L^2$ is defined as the best $L^2$-estimate, given a $\sigma$-subalgebra and a degree of model uncertainty.  When the random vector represents the payoff of derivative security in a complete financial market, its $\rrrr$-conditioning with respect to the risk-neutral measure is interpreted as its risk-averse value.  
The optimization problem defining the optimization $\rrrr$-conditioning is shown to be well-posed.  We show that the $\rrrr$-conditioning operators can be used to approximate a large class of sublinear expectations to arbitrary precision.  We then introduce a novel numerical algorithm for computing the $\rrrr$-conditioning.  This algorithm is shown to be strongly convergent.

Implementations are used to compare the risk-averse value of a Vanilla option to its traditional risk-neutral value, within the Black-Scholes-Merton framework.  
Concrete connections to robust finance, sensitivity analysis, and high-dimensional estimation are all treated in this paper.  
\end{abstract}

Risk-Averse Valuation, High-Dimensional Probability, Sparse Conditional Expectation, Sparisty.

\noindent
\let\thefootnote\relax\footnotetext{This research was supported by the ETH Z\"{u}rich Foundation.}

\section{Introduction}
Let $(\Omega,\ffff,\{\ffff[t]\}_{t\geq 0},\pp)$ be a filtered probability space, and let $X_t$ be an $\ffff[t]$-adapted semi-martingale with values in $\rr$.  When considering financial applications, we interpret $X_t$ as a price process and assume that the \cite[No Free Lunch with Vanishing Risk (NFLVR)]{delbaen1994general} formulation of market completeness holds.  Recall that this is the same as the existence of a unique equivalent local martingale measure (ELMM) for the discounted price process $e^{-rt}X_t$.  We denote the ELMM by $\qq$ and assume that $r\geq 0$.

In classical risk-neutral pricing theory, a financial derivative on $X_t$ with maturity $T>0$, and a Borel-measurable payoff function $f:\rrd\rightarrow \rr$ is priced according to the risk-neutral pricing formula
\begin{equation}
V(T,%
f)\triangleq e^{-rT}\ee{f(X_T)}[
\ffff[0]%
][\qq]
\label{RN_1}
.
\end{equation}
Let us assume that $f(X_t)$ and $X_t$ are square-integrable under $\qq$, that is, $f(X_t),X_t \in\Lp{\ffff[]}[]$.  Then~\eqref{RN_1} can be expressed as the solution to the following quadratic optimization problem
\begin{equation}
V(T,%
f)=\arginf{Z \in L^2(\ffff[0])} \eemes{\|f(X_T)-Z\|^2}[\qq]
\label{RN}
\tag{RN}
.
\end{equation}

If the model does not accurately reflect reality, then the estimate of the derivative security's payoff may be quite poor, since many factors driving financial derivative's the price may be ignored by the model, and consequentially they are overlooked in the pricing problem~\eqref{RN}.  The robust finance literature proposes a solution to this issue.  For example, in \cite{coquet2002filtration,peng2010nonlinear,soner2011martingale,cohen2011sublinear}, a family of plausible alternative probability measures to $\pp$, denoted by $\qqqq$, is used to quantify the best estimate under uncertainty.  Under appropriate assumptions on $X_t$ and on $\qqqq$, ~\eqref{RN} generalizes to either
\begin{equation*}
\esssup{Q \in \qqqq} \arginf{Z \in \mathcal{X}}\eemes{\|f(X_T)-Z\|^2}[Q] - B(Q)
,
\end{equation*}
or to
\begin{equation}
\arginf{Z \in \mathcal{X}} \sup_{Q \in \qqqq}\eemes{\|f(X_T)-Z\|^2}[Q] - B(Q)
\label{intro_eq_RF}
,
\end{equation}  
where $B$ is a function, capturing the modeler's bias towards certain measures in $\qqqq$, and $\mathcal{X}$ is a suitable subset of the set of $(\ffff[0],\mathcal{B}(\rr))$-measurable functions.
The two formulations are equivalent under appropriate conditions on $\qqqq$ and $B$.

For the moment, let us think of uncertainty as risk-aversion; a perspective which is justified in later portions of this paper.  If we do so, then problem~\eqref{RN} and problem~\eqref{intro_eq_RF} would sit on the opposite sides of the risk-aversion spectrum.  An investor pricing, according to ~\eqref{RN}, places complete confidence in their model, which reflects them having a minimal level of risk-aversion.  On the other hand, if the investor prices according to~\eqref{intro_eq_RF}, then they would be placing essentially no confidence in their model, and so they would be expressing an extreme amount of risk-aversion into their price.  

In reality, however, most market participants tend to hold an intermediate degree of risk-aversion.  The incorporation of risk-aversion in mathematical finance dates back to the advent of modern portfolio theory in \cite{markowitz1952portfolio}, where the author characterized portfolios offering maximal returns for a pre-determined risk-aversion level.  For the author of \cite{markowitz1952portfolio}, risk-aversion is quantified by the variance of a portfolio's returns.  Markowitz's approach has since been modernized in \cite{kull2014portfolio,gilli2002global,ciliberti2007feasibility} with the risk-aversion constraint being reformulated as a constraint on a portfolio's risk, as quantified either by the value-at-risk or by the expected shortfall of the portfolio's returns.  

The incorporation of risk-aversion, into the estimation procedure, as expressed by risk-measure constraints, has since been successfully undertaken in several other branches of mathematical finance.  For example, \cite{jiao2017hedging} built on aspects of the hedging literature, founded in  \cite{schweizer1988hedging,schweizer1991option,schweizer1995variance,delbaen2002exponential,schweizer2010mean,jeanblanc2012mean}, by describing a procedure for finding the most inexpensive portfolio meeting a set of stochastic benchmarks which are required to be bounded by a set of risk measures.

We build on the success of these methods, by viewing the loss function in~\eqref{intro_eq_RF}, not as the objective defining the optimization problem itself, but instead as a constraint to the original risk-neutral pricing problem~\eqref{RN}.  Explicitly, we begin our analysis by refining~\eqref{RN} as
\begin{equation}
\begin{aligned}
&\arginf{Z \in \Lp{\ffff[0]}[]}[1]     \eemes{
\left(f(X_T)-Z\right)^2
}[\qq]
\\
&\mbox{subject to } 
\left(
\sup_{Q \in \qqqq}\eemes{
\left(f(X_T)-Z\right)
^2
}[Q]
- B(Q)
\right)
\leq M
.
\end{aligned}
\label{RA}
\tag{RA}
\end{equation}
We will see that~\eqref{RA} can be computed using efficient convex analytic algorithms; however, these algorithms do not apply to~\eqref{intro_eq_RF}.  

An added benefit of problem~\eqref{RA} over~\eqref{intro_eq_RF} is that the objective measure $\pp$ and the risk-neutral measure $\qq$ are both used.  %
In this way, Problem~\eqref{RA} incorporates many of the desirable qualities of the Problem~\eqref{intro_eq_RF} without most of its intractability.  

Another major difference between~\eqref{RA} and~\eqref{intro_eq_RF} is that the solution operator is typically a highly non-linear mapping.  Therefore the sublinear expectation theory of \cite{coquet2002filtration,cohen2011sublinear,cohenqs} and the convex risk-measure literature of \cite{gao2018c,cheridito2008dual,KainaCRMs,filipovic2007convex,delbaen1994general,mcneil2005quantitative} cannot be used directly in the study of problem~\eqref{RA}.  Therefore, we establish a new generalized conditional expectation theory, falling somewhere in between that of the classical conditional expectation theory and the sublinear expectation theory.

\subsection{Additional Features and The General Problem}%
Though problem~\eqref{RA} addresses many of the issues of the former two estimation problems, it is incomparable with three of the focuses of modern mathematical finance.  The first of these is the crudeness of the quadratic utility function.  One immediate issue with this utility function is that it is unable to distinguish between gains and losses from misestimation.  Authors such as \cite{biagini2008unified,cheridito2008dual,gao2018c} have been developing frameworks that are capable of working with a general utility function.  We also consider these features in the general formulation of~\eqref{RA}.  

Secondly, contemporary financial data-sets are typically very high-dimensional.  Authors such as \cite{breymann2003dependence,dias2004dynamic,donoho2000high,tibshirani1996regression} have all made efforts to incorporate such features into their modeling approach.  We require that high-dimensional data can be encoded into it.  To incorporate this feature into our problem, we assume that all processes and random-elements take values in a separable Hilbert space $\hhh$, which we typically intuit to be infinite-dimensional.  Therefore, the optimization of~\eqref{RA} is instead performed over the Bochner-Lebesgue space $\Lp{\ffff[]}$, which is itself a Hilbert space (see \cite{evans1990weak} for details on infinite-dimensional integration).  

Thirdly, following the financial crisis, the regulatory requirements of the Basel accords \cite{balin2008basel,basel2004international,basel2013fundamental} have placed restrictions on how market participants can transact.  We incorporate this additional constraint into problem~\eqref{RA} by restricting the optimization over $\Lp{\ffff[]}$ to a suitable subset $\Phi$.  We interpret the elements of $\Phi$ as estimators, which posses additional desirable features.  An example of such features is, meeting other regulatory requirements or admitting a sparse representation.

Before stating the general form of~\eqref{RA} studied in this paper, which incorporates all these features we recall that in \cite{KainaCRMs}, it is shown that, if a suitable choice of alternative probability measures $\qqqq$ and a bias function $B$ are made, then the map from $\Lp{\ffff[]}$ to $\rrext$ defined by
\begin{equation}
X\mapsto \sup_{Q \in \qqqq}\eemes{
Z
}[Q]
- B(Q)
,
\label{intro_rob_rep}
\end{equation}
is equivalent to specifying a proper, lower semi-continuous (lsc), convex risk-measure on $\Lp{\ffff[]}[]$.  We denote this risk-measure by $\rho$.  The representation of $\rho$ given by~\eqref{intro_rob_rep}, which is called a robust representation, makes up part of a large literature, relating risk-measures to robust finance, see \cite{delbaen2002coherent,cheridito2008dual,acciaio2011dynamic} for example.  

The formulation of our problem is more general and does not refer to time, or the payoff function $f$.  Instead, it considers random vectors in $\Lp{\ffff}$ and targets of the form $f(X_T)$ are only specified when dealing with explicit financial situations.  The main objective of this paper is to study the following generalization of Problems~\eqref{RN} and~\eqref{RA}, 
\begin{equation}
\begin{aligned}
&\arginf{Z \in \Lp{\gggg}}[1]     \eemes{
\|Z-X\|^2
}[\qq]\\
&\mbox{subject to } \rho(\utility(X-Z))\leq M,\\
&\mbox{subject to } Z \in \Phi,
\end{aligned}
\label{PG}
\tag{PG}
\end{equation}
where %
$\gggg$ is a sub-$\sigma$-algebra of $\ffff$, $\utility$ is a suitable utility-like function%
, and $\ee{\cdot}$ denotes the expectation with respect to the measure $\pp$%
.  All these quantities will be described rigorously in Section $2$.

\subsection{Objectives of the Paper}
The central objectives of this paper are first to give suitably general conditions under which~\eqref{PG} is well-posed, and to describe a procedure which can (strongly) approximate it to arbitrary precision.  The secondary goals are to study the general properties of the solution operator to~\eqref{PG}, as well as connections to robust finance, and other areas of contemporary mathematical finance.

\subsection{Organization of Paper}\label{ss_Organization_Paper}
Section~\ref{s_Bgrnd} briefly overviews the necessary background in risk-measures and some of the relevant convex analysis tools used in the rest of this paper.  Section~\ref{s_RACE} establishes the well-posedness of~\eqref{PG} as well as the basic risk-reduction properties of the solution to~\eqref{RA}.  Section~\ref{s_sol_Uncertainty} reviews and uses the theory of $\Gamma$-convergence to show that, in many cases, the solutions of~\ref{RA} can be used to asymptotically solve~\ref{intro_eq_RF} to arbitrary precision.

In Section~\ref{s_compute}, we introduce a novel numerical scheme for computing~\eqref{PG}, and we establish its strong convergence.  
Section~\ref{s_rel_oLit}, makes connections between the solution to~\eqref{PG} and robust finance.  Specifically, a robust representation is derived, we provide connections to regular extensions of convex risk-measures on $\Lp{\ffff[]}[][p]$ spaces, as well as connections to sensitivity analysis, and a link to $\gggg$-consistent sublinear expectations.  We explore links to high-dimensional probability theory by showing that the set $\Phi$ can be used to obtain sparse estimators.

An implementation comparing the risk-averse value of a Vanilla call option to their risk-neutral value, within the framework of~\cite{black1973pricing}, is considered in Section~\ref{s_example_RA_Vanilla}.

\section{Background}\label{s_Bgrnd}
The subjectivity of risk-aversion in the mathematical finance literature, dates all the way back to~\cite{markowitz1952portfolio}, where different investor-dependent risk-appetites were incorporated into the optimal portfolios.  Before formalizing a perspective on risk, in a way which is appropriate with problem~\eqref{PG}, we take a moment to review some of the relevant risk-measure literature.
\subsection{Risk-Measures}\label{s_RMs}
A \textit{(possibly non-finite)} proper, convex risk-measure on $\Lp{\ffff[]}[]$ is defined to be a function $\rho:\Lp{\ffff[]}\rightarrow \rrext$ which is monotone, cash-invariant, and convex; that is for any $X,Y \in \Lp{\ffff[]}$, $k \in \rr$, and $\alpha \in [0,1]$, $\rho$ satisfies:
\begin{enumerate}[(i)]
	\item If $X\leq Y$ $\pp$-a.s. then $\rho(X)\leq \rho(Y)$,
	\item $\rho(X+ k)= \rho(X)-k$,
	\item $\rho(\alpha X + (1-\alpha)Y)\leq \alpha \rho(X) + (1-\alpha)\rho(Y)$,
	\item There exists $Z \in \Lp{\ffff[]}[]$ such that $\rho(Z)\in \rr$.
\end{enumerate}
\begin{remark}[Convention: Upper Addition]\label{rem_upp_add}
	For the remainder of this paper, follow the convention of \cite{rockafellar2009variational}, that $\infty + r=\infty$ for $r \in \rrbar$; where $\rrbar \triangleq [-\infty,\infty] $.  Some authors, such as \cite{stromberg1994study}, call this upper-addition to contrast the alternative convention that $-\infty +\infty = -\infty$, which we will not follow in this paper.  
\end{remark}

As mentioned in~\eqref{intro_rob_rep}, a convex risk-measure admits the following robust representations due to \cite{biagini2006continuity},
\begin{equation}
\rho(X) = \sup_{Q \in \qqqq} \eemes{-X}[Q] - \rho^{\star}\left(\frac{dQ}{d\pp}\right)
\label{eq_robust_representation_Lp_risk_measures}
,
\end{equation}
where $\rho^{\star}$ is the Fenchel-Moreau (convex) conjugate of $\rho$ on $\Lp{\ffff[]}[]$, $\qqqq$ is a non-empty subset of 
$\left\{
Q \in \mathscr{P}(\Omega,\ffff[]) : Q \ll \pp \mbox{ and } \frac{dQ}{d\pp} \in \Lp{\ffff[]}[]
\right\}$, and $\mathscr{P}(\Omega,\ffff[])$ is the set of probability measure on $(\Omega,\ffff[])$.

The assumption of lower semi-continuity is necessary if the representation~\eqref{eq_robust_representation_Lp_risk_measures} is to hold.  In \citep[Lemma 2.3]{filipovic2007convex}, it is shown that the lower semi-continuity of $\rho$ is equivalent to the (strong) closedness of its acceptance set $\aaaa_{\rho}$, defined by
\begin{align}
\label{eq_defn_acceptance_set}
\aaaa_{\rho}\triangleq &\left\{
Z \in \Lp{\ffff[]}[] : \rho(Z)\leq 0
\right\}
.
\end{align}

In general however, the (strong) continuity of $\rho$ is only guaranteed, by the \citep[Extended Namioka Theorem]{biagini2006continuity}, on the interior of its domain\footnote{
The domain of a convex function $f:\Lp{\ffff[]}[]\rightarrow \rrext$, denoted $\operatorname{dom}(f)$, is defined to be the set
$
\operatorname{dom}\left(\utility\right)\triangleq \left\{
Z \in \Lp{\ffff}: \utility(Z)\neq \infty
\right\}
.
$
}.  In particular, as concluded in \citep{KainaCRMs}, a finite-valued risk-measure is continuous on all of $\Lp{\ffff[]}[]$.  

The algorithmic portion of this paper will require the convex risk-measure defining~\eqref{PG} to be of higher regularity; specifically, we will require that $\rho$ be G\^{a}teaux differentiable on the interior of its domain.  From \cite{cheridito2008dual}, it can be seen that if the subdifferential of $\rho$ at $X \in \Lp{\ffff[]}[]$, defined by
$$
\partial \rho(X) \triangleq \left\{
Z \in \Lp{\ffff[]} :
(\forall Y \in \Lp{\ffff[]}[])\;
\rho(X +Y ) - \rho(X)\geq \ee{YZ} 
\right\}
,
$$
is single-valued, then unique element in $\partial \rho(X)$ is the G\^{a}teaux derivative of $\rho$ at $X$.  We denote the G\^{a}teaux derivative of $\rho$ at $X$ by $\nabla \rho(X)$, and recall that the G\^{a}teaux derivative can be equivalently defined by a weak calculus, if
$$
\ee{Y\nabla\rho(X)} = \lim\limits_{\epsilon \downarrow 0} 
\frac{
	\rho(X+ \epsilon Y) - \rho(X)
}{\epsilon}
,
$$
holds for all $Y \in \Lp{\ffff[]}[]$.  We note that by \cite[Proposition 17.39]{ConvexMonoCMS}, when $\rho$ is G\^{a}teaux differentiable, the selection taking $X$ to its G\^{a}teaux derivative is not only well-defined, but it is also strong-to-weak continuous on $\operatorname{int}(\operatorname{dom}(\rho))$.  

The implementation portion of this paper will prefer a certain analytic tractability of the G\^{a}teaux derivative.  In \cite{cheridito2009risk} a class of risk-measures with a particularly tangible G\^{a}teaux derivative were introduced.  The construction proceeds as follows.  Consider a monotonically decreasing, convex map from $V:\Lp{\ffff[]}[]$ to $\rrext$, for which the following set is single-valued
\begin{equation}
\arginf{k \in \rr}V(k-X)-k
\label{eq_inf_exist}
.
\end{equation}
Denote the (unique) selection, taking an element $X \in \Lp{\ffff[]}[]$ into the set~\eqref{eq_inf_exist} by $V_X$.  The cash-additive hull of $\rho$ is the largest real-valued convex risk-measure map on $\Lp{\ffff[]}[]$ dominated by $V$.  It is typically denoted by $\rho_V$ and is defined through
\begin{equation}
\begin{aligned}
\rho_V(X)\triangleq & \left[
	V\left(V_X - X
	\right)-V_X
\right] - V\left(V_0\right)
\end{aligned}
\label{eq_cash_additive}
.
\end{equation}

Our interest with cash-additive hulls stems from \citep[Proposition 7.1]{cheridito2008dual}.  Therein, it was shown that if $V$ is G\^{a}teaux differentiable at $V_X-X$, then $\rho_V$ is G\^{a}teaux differentiable at $X$.  Moreover, its G\^{a}teaux derivative admits the following convenient analytic form
\begin{equation}
\nabla\rho_V(X)= -\nabla V\left(V_X-X\right)
\label{eq_cah_simple}
.
\end{equation}
The next example will be central to the implementations considered herein.  
\begin{ex}[Quadratic Risk-Measure]\label{lem_Qudratic_risk}
	Consider the case where \\$V(X)\triangleq \frac1{2}\ee{\left\|X\right\|^2}$.  We will denote its cash-additive hull by $\rho_2$.  Using standard calculus to compute the minimizer $V_X$ of~\eqref{eq_inf_exist}, we find that
	\begin{align}
	\rho_2(Z)&= \ee{
		\left(
		\ee{Z}+\frac1{2} -Z
		\right)^2 
	} - \left[\ee{Z}+\frac1{2}\right] -\frac1{4}
	\\
	V_Z&= %
	\ee{Z} + \frac1{2} 
	.
	\end{align}
	
	In \citep[Section 8.2]{cheridito2008dual} $\rho_2$ was shown to be G\^{a}teaux differentiable on $\Lp{\ffff[]}[]$.  Therefore,~\eqref{eq_cah_simple} implies that its G\^{a}teaux derivative is given by
	\begin{equation}
	\nabla \rho_2(Z) %
	= Z -\left(\ee{Z} +\frac1{2}\right)
	.
	\label{eq_simplific}
	\end{equation}
	
	Furthermore,~\citep[equation 8.5]{cheridito2008dual} and~\eqref{eq_robust_representation_Lp_risk_measures} imply that $\rho_2$ can be understood through the following robust representation
	\begin{equation}
	\begin{aligned}
	\rho_2(Z)&= \sup_{\qq \in \qqqq}\ee{-Z} - \ee{\left(\frac{d\qq}{d\pp}\right)^2}
	,
	\\
	\qqqq &\triangleq \left\{
	\qq\ll \pp : \ee{\frac{d\qq}{d\pp}}=1, \frac{d\qq}{d\pp}\geq 0, \frac{d\qq}{d\pp}\in \Lp{\ffff[]}[]
	\right\}
	.
	\end{aligned}
	\end{equation}
\end{ex}

In the next section, we review some of the relevant background on the infimal convolution operator by making contrasts with functional analysis and measure theory.  
\subsection{The Infimal Convolution Operation}\label{rem_inf_conv}
The results subsection are formulated in $\Lp{\ffff[]}$, though they may be set in a more general context.

Any conditional expectation operator in $\Lp{\ffff[]}$ is a metric projection onto the the linear subspace $\Lp{\gggg}$.  In general, projections can be defined on any non-empty, closed, convex subset $C$ of $\Lp{\ffff[]}$ by
\begin{equation}
\Pi_C(X)\triangleq \arginf{Z \in C}[1] \ee{\|X-Z\|^2},
\label{eq_proj}
\end{equation}
standard results in Hilbert space analysis show that the map $X\mapsto \Pi_C(X)$ is indeed well-posed.  

Using the indicator function of the set $C$, defined by
$$
\iota_C(Z)\triangleq \begin{cases}
0 & : Z \in C\\
\infty & : Z \not\in C
\end{cases}
,
$$
the projection operator of~\eqref{eq_proj} may be rewritten as
\begin{equation}
\Pi_C(X)\triangleq \arginf{Z \in \Lp{\ffff[]}}[1] \frac1{2}\ee{\|X-Z\|^2} + \iota_{C}(Z).
\label{eq_proj_alt}
\end{equation}

It was noticed in \cite{moreau1965proximite}, that replacing the indicator function $\iota_C$ by any proper, convex, and lsc function $f$ on a Hilbert space (which in our context is $\Lp{\ffff[]})$, many of the desirable \textit{"projection-like"} properties of $\Pi_C$ are retained.  This lead the author to consider the following generalization of a metric projection%
\begin{equation}
\operatorname{Prox}_{f}(X)\triangleq
\arginf{Z \in \Lp{\ffff[]}}[1] \frac1{2}\ee{\|X-Z\|^2} + f(Z)
\label{eq_Prox_defn}
.
\end{equation}
It was proven in~\cite{moreau1965proximite} that the multi-function $\operatorname{Prox}_{f}$ is in-fact single-valued and $1$-Lipschitz.  The author shows that for every $\eta >0$, the function $f^{\eta}$, called the Moreau-Yoshida envelope of $f$ with parameter $\eta$, defined by 
\begin{equation}
\begin{aligned}
{f}^{\eta}(X)\triangleq &
\inf_{Z \in \Lp{\ffff[]}} \frac1{2}\ee{\|X-Z\|^2} + f(Z)
,
\end{aligned}
\label{eq_Moreau_defn}
\end{equation}
is lsc, convex, real-valued, locally Lipschitz, and is in-fact Fr\'{e}chet differentiable.  Many applications of this map have been seen, either implicitly or explicitly, in statistical learning (see \cite{tibshirani1996regression,zou2006sparse} for example).  

\begin{ex}[{\citep[Example 12.21, Corollary 12.23]{ConvexMonoCMS}}]\label{ex_MYenv}
	The Moreau-Yoshida envelope of $\iota_{C}$ can be computed to be
	$$
	\iota_{C}^{\eta}(X) = \frac1{2\eta}d_C^2(X)%
	,
	$$
	where $d_C(X)$ represents the shortest distance between the random-vector $X$ and the set $C$.  Moreover, the Moreau-Yoshida envelope $\iota_C^{\frac1{2}}$ can be used to compute the projection $\Pi_C(X)$ through
	\begin{equation}
	\Pi_C(X)=X-\frac1{2}\nabla d_C^2(X)
	,
	\label{eq_motiv_gen_relat_Prox_MY_env}
	\end{equation}
	where $\nabla d_C^2$ denoted the Fr\'{e}chet derivative of $d_C^2$.   
\end{ex}
The relationship~\eqref{eq_motiv_gen_relat_Prox_MY_env} is not coincidental and extends to proximal mapping of any proper, lsc, convex function $f$.  In general, $\Prox{f}$ can be computed using $f^{\eta}$ through
\begin{equation}
\begin{aligned}
\nabla{f}^{\eta}(X) = & \frac1{\eta}\left(X%
-\operatorname{Prox}_{f}(X)\right)
.
\end{aligned}
\label{eq_Moreau_defn_relation_prox}
\end{equation}

Replacing the functional $X\mapsto \frac1{2}\ee{\|X\|^2}$ by any proper, convex, lsc function $g$ in the definition of the Moreau-Yoshida envelope, we obtain the infimal convolution of $f$ and $g$, defined by
$$
f\square g (X)\triangleq \inf_{Z \in \Lp{\ffff[]}} g(X-Z) + f(Z)
.
$$
Geometrically, the infimal convolution is  the largest extended real-valued function whose epigraph contains the sum of epigraphs\footnote{The epigraph of a function, is the set of all values lying on or above that function's graph.  It encodes many of the properties of proper convex functions.  } of $f$ and
$g$.  Analogously to the single-valuedness of the proximal mapping operator, the infimal convolution is said to be \textit{exact} at $X$, written $f\boxdot g$, if there exists a unique element $X^{\star}\in \Lp{\ffff[]}$ satisfying
$$
f \square g (X) = g(X-X^{\star}) + f(X^{\star})
.
$$
However, unlike $\operatorname{Prox}_f$, the infimal convolution is not guaranteed to be exact.

\section{{$\rrrr$-}Conditioning}\label{s_RACE}
We are now in place to provide a formalization of an investor's perspective on risk or partial uncertainty towards their model.  
\subsection{Risk-Perspective}\label{ss_rp}
\begin{defn}[Risk-Perspective]\label{defn_perspective}
	A risk-perspective $\rrrr$, is a quintuple\\ $(\gggg,\rho,\utility,\Phi,M)$ comprised of a:
	\begin{enumerate}[(i)]
		\item \textbf{$\sigma$-Subalgebra:} $\sigma$-sub-algebra $\gggg$ of $\ffff[]$,
		\item \textbf{Risk-Measure:} Proper, convex risk-measure $\rho$, on $\Lp{\ffff[]}$, with a non-empty and (strongly) closed acceptance set $\aaaa_{\rho}\triangleq \left\{
		Z \in \Lp{\ffff[]}[]: \rho(Z)\leq 0
		\right\}$,
		\item \textbf{Compatible Utility Operator:} A function $\utility:\Lp{\ffff[]}\rightarrow %
		\Lp{\gggg}$, such that $\riskutility\triangleq \rho\circ \utility$ is proper, convex, and (strongly) lower semi-continuous, and satisfies
		\begin{equation*}
		-\infty<\inf_{Z \in \Lp{\gggg}} \riskutility(Z)
		,
		\end{equation*}
		\item \textbf{Feature Set:} A non-empty, convex, and (strongly) closed subset $\Phi$ of $\Lp{\ffff[]}$, such that
		\begin{equation*}
		\Lp{\gggg[]}\cap \Phi\cap \aaaa_{\rho} \cap \operatorname{dom}\left(\riskutility\right)\neq \emptyset
		,
		\end{equation*}  
		where $\operatorname{dom}\left(\riskutility\right)\triangleq \left\{
		Z \in \Lp{\ffff[]} : \riskutility(Z)<\infty
		\right\}$. 
		\item \textbf{Risk-Aversion Level:} An extended real number $M\in (0,\infty]$.  
	\end{enumerate}
\end{defn}
To make things concrete, let us consider a few examples of compatible utility operators and feature sets before moving on.  
\begin{ex}[Compatible Utility Operator: %
	Quadratic Loss]\label{ex_dist_target}
	The map \\
	$Z\mapsto -\|%
	Z\|^2$ is smooth, concave, and non-increasing.  Since $\rho$ is convex an non-increasing then $\riskutility$ is convex on its domain.  Moreover, if $\rho$ is continuous, then the smoothness of $\|\cdot\|^2_{\hhh}$ implies that $\riskutility$ is continuous on its domain.  
\end{ex}
More generally, we have the following example.  
\begin{ex}[Compatible Utility Operator: Non-Euclidean Utility]\label{ex_dist_target_gen}
	Fix $X\in \Lp{\ffff[]}$ and let $f:\hhh\rightarrow [0,\infty)$ be a decreasing, smooth, convex function satisfying
	$$
	f(x)=0 \Leftrightarrow x=0.
	$$
	
	The map    $Z\mapsto f(X-Z)$ defines a compatible utility operator.  To see this, note that since $\rho$ is convex an non-increasing then $\riskutility$ is convex on its domain.  Moreover, if $\rho$ is continuous, then the smoothness of $f$ implies that $\riskutility$ is continuous on its domain.  
\end{ex}
\begin{ex}[Compatible Utility Operator: Thresheld Aggregation]\label{ex_affine}
	Suppose $\hhh=\rrd$ and let $L\in \rr$.  In \cite{billio2012econometric,acharya2017measuring,aminiy2013systemic}, measures of systemic risk are described through aggregation.  This is accomplishes through the use of an aggregation function $g$ from $\rrd$ to $\rr$ such that $\rho\circ g$ is itself a convex (resp. coherent) risk-measure.  Define $\utility_L$ to be
	$$
	\utility_L(Z)\triangleq \max\left\{
	g(Z),L
	\right\}.
	$$
	If $g$ is itself lsc, then $\utility_L$ must also be lsc.  
	
	The monotonicity of $\rho$ implies that the minimum value attainable by $\rho\circ \utility_L$ is $-L$.  Therefore, if $\rho\circ \utility_L$ is itself convex, then $\utility_L$ is a compatible utility function.  
\end{ex}
\begin{ex}[Feature Set: Classical Conditional Expectation]\label{ex_admis}
	The set $\Lp{\ffff[]}$ satisfies the definition of an feature set.  Moreover, if $M=\infty$, then~\eqref{PG} is precisely the classical vector-valued definition of conditional expectation.  If furthermore $\hhh=\rr$,~\eqref{PG} reduces to~\eqref{RN}.  
\end{ex}
\begin{ex}[Feature Set: Acceptance Set of A Risk-Measure with Aggregation]\label{ex_accept_RM}
	Set $\hhh=\rrd$, $m\geq 0$.  Let $g:\rrd\rightarrow \rr$ be an continuous function, $R:\Lp{\ffff[]}[]\rightarrow \rrext$ be a convex, lsc, risk measure, and suppose that the composition $R\circ g$ is convex.  Define the shifted aggregated acceptance set $\aaaa_{R}^{g,m}$ by
	\begin{equation}
	\aaaa_{R}^{g,m}\triangleq \left\{
	X \in \Lp{\ffff[]}[] 
	:
	R\left(g(X)\right)\leq m
	\right\} = (R\circ g)^{-1}\left[(-\infty,m]\right]
	.
	\label{eq_acceptance_shifted}
	\end{equation}
	Since $R$ is lsc, and $\aaaa_{R}^{g,m}$ is a level set, then it is closed.  Likewise, the convexity of $R$ and that of the half-line $(-\infty,m]$ imply that $\aaaa_{R}^m$ is itself convex.  Hence $\aaaa_{R}^{g,m}$ is an feature set.  
	
	In the case where $\rrd=\rr$ and $g(x)=x$, Problem~\eqref{RA}, may therefore be rewritten as
	\begin{equation}
	\begin{aligned}
	&\arginf{Z \in \Lp{\ffff[0]}[]}[1]     \eemes{
		\left(f(X_T)-Z\right)^2
	}[\qq]
	\\
	\mbox{subject to } &
	\rho\left(
	\utility\left(f(X_T)-Z\right)
	\right)
	\leq M
	\\
	& R\left(
	Z
	\right)
	\leq m
	.
	\end{aligned}
	\label{eq_constrianed}
	\end{equation}
	A solution to Problem~\eqref{eq_constrianed}, if it exists, seeks to find an estimator of the price $f(X_T)$, which is itself both low-risk and has a low risk of making an estimation error, as quantified by $\utility$.  
\end{ex}

\begin{ex}[Feature Set: Thresheld Variance]\label{ex_Prespecified_Variance_Set}
	Let $\hhh=\rrd$.  In the early portfolio optimization literature of \cite{markowitz1952portfolio}, risk is quantified by variance.  Given a maximum variance level $\Sigma>0$, the maximum returns portfolio with variance $\Sigma$, is a random vector in the set
	$$
	\Phi_{\Sigma}^2\triangleq \left\{
	Z \in \Lp{\ffff[]} :
	\sqrt{\ee{\|Z\|^2}}\leq \Sigma
	\right\}
	.
	$$  
	
	Note that $\Phi_{\Sigma}^2$ is a ball in Bochner-Lebesgue space $\Lp{\ffff[]}[d]$ about $0$.  Therefore, the continuity and linearity of the norm on that space ensure that $\Phi_{\Sigma}^2$ is a non-empty closed convex subset of $\Lp{\ffff[]}[d]$.
	Hence $\Phi_{\Sigma}^2$ is an feature set.  
\end{ex}
\begin{ex}[Feature Set: Sparse Estimators]\label{ex_admis_sparse}
	Let $\Sigma>0$.  Define the subset $\Phi^1_{\Sigma}$ of $\Lp{\ffff[]}[d]$ by
	$$
	\Phi^1_{\Sigma}\triangleq \left\{
	X \in \Lp{\ffff[]}[d] :\,
	\operatorname{ess\,supp}(X)
	\subseteq
	\overline{\operatorname{Ball}_1(0;\Sigma)}%
	\right\}
	$$
	where $\operatorname{Ball}_1(0;\Sigma)$ denotes the closed $\ell^1$-ball in $\rrd[d]$ and $\operatorname{ess\,\sup}$ denotes the essential support of a $D$-dimensional random-vector $X$ is defined by
	$$
	\operatorname{ess\,supp}(X) := \bigcap \left\{
	K \subseteq \rrd[D] : \pp\left(X \in K\right)=1  \mbox{ and } K \mbox{ is closed}
	\right\}
	.
	$$
	We will show that $\Phi_{\Sigma}^1$ is a feature set.  
	
	We use the following construction twice in this example.  \hfill\\
	\textit{
		Let $\Gamma$ be a countable subset of $\Phi^1_{\Sigma}$.  For any $X\in \Lp{\ffff[]}[d]$ define its modification $X^{\Gamma}$ by
		$$
		X^{\Gamma}(\omega)\triangleq \begin{cases}
		X(\omega) & : (\forall Y \in \Gamma) \; Y(\omega) \in \overline{\operatorname{Ball}_1(0;\Sigma)}\\
		0 &: else
		\end{cases}
		.
		$$
		Since $\Gamma$ is countable, then the set $\bigcup_{Y \in \Gamma}\{\omega \in \Omega: Y(\omega) \not\in K\}$ is a countable union of probability $0$ subsets of $\Omega$; hence is itself of probability $0$.  Therefore the $X^{\Gamma}$ is indeed a well-defined modification of $X$.%
	}
	The modification $X^{\Gamma}$ will allow us to argue $\omega$-wise.  
	
	First, we establish the convexity of $\Phi^1_{\Sigma}$.  Let $X,Y \in \Phi^1_{\Sigma}$, $\lambda \in [0,1]$, and set $\Gamma\triangleq \left\{X,Y\right\}$.  
	then the convexity of $\overline{\operatorname{Ball}_1(0;\Sigma)}$ implies that
	$$
	\lambda{X}^{\Gamma}(\omega) +(1-\lambda){Y}^{\Gamma}(\omega) \in \overline{\operatorname{Ball}_1(0;\Sigma)};
	$$
	hence $\lambda X + (1-\lambda)Y \in \overline{\operatorname{Ball}_1(0;\Sigma)}$ $\pp$-a.s. 
	
	We now establish the closedness of $\Phi^1_{\Sigma}$ in $\Lp{\ffff[]}[d]$.  Let $\left\{X_n\right\}_{n \in \nn}$ be a sequence in $\Phi^1_{\Sigma}$ having a limit in $\Lp{\ffff[]}[d]$; we denote this limit by $X$.  Since convergence in $\Lp{\ffff[]}[d]$ implies convergence in probability to $X$.  Since $\left\{X_{n}\right\}_{k \in \nn}$ to $X$ in probability, there exists a subsequence $\left\{X_{n_k}\right\}_{k \in \nn}$ converging to $X$~$\pp$-a.s.  
	
	Since the set $\Gamma \triangleq \{X_n\}_{n \in \nn}\cup \{X\}$ is countable, the modification $X^{\Gamma}$ as well as the modified subsequence $\{X_{n_k}^{\Gamma}\}_{n \in \nn}$ are both well-defined.  The subsequence takes values in the subset $\mathcal{X}$ of $K$ defined by
	$$
	\mathcal{X}\triangleq \bigcup_{k \in \mathbb{N},\omega\in\Omega} \tilde{X}_{n_k}(\omega).  
	$$
	Moreover, $X^{\Gamma}$ must take values in the $\rrd$-closure of $\mathcal{X}$.  However, since $\operatorname{Ball}_1(0;\Sigma)$ is a closed subset of $\rrd$, then $\overline{\mathcal{X}}\subseteq \operatorname{Ball}_1(0;\Sigma)$.  Hence $\operatorname{supp}\left(X^{\Gamma}\right)\subseteq \operatorname{Ball}_1(0;\Sigma)$, which implies that $\operatorname{ess\,supp}(X)\subseteq \operatorname{Ball}_1(0;\Sigma)$.  Therefore $X \in \Phi^1_{\Sigma}$, thus $\Phi^1_{\Sigma}$ is closed in $\Lp{\ffff[]}[d]$.  Hence, $\Phi^1_{\Sigma}$ is a well-defined feature set.  
\end{ex}
\subsection{$\rrrr$-Conditioning}
Next, we study the solution operator to~\eqref{PG}, under the following assumption on the compatibility between the geometry of $\Phi$ and that of $\riskutility$.  
\begin{ass}\label{ass_reg_cone}\hfill
\begin{enumerate}[(i)]
\item There exists a convex cone $S\subseteq \Lp{\ffff[]}$ containing the optimizers of problems~\eqref{PG} and~\eqref{thrm_existsnce_uniqueness_eq_defining_optimization_problem}, below.  
\item There exists $M^{\star}\in (0,\infty)$, such that for every $M>
M^{\star}
$, the constraint set $\Phi$ is such that there exists $\hat{Z}\in \Phi\cap \Lp{\gggg}$ such that
\item \begin{enumerate}[(a)]
	\item $\operatorname{cone}\left(\Phi-\hat{Z}\right)$ is norm-closed,
	\item $%
	\hat{Z}%
	\in -\operatorname{int}\left(\left\{
	Y \in \Lp{\ffff[]}: \riskutility(Y)\leq M
	\right\}
	\right)$.  
\end{enumerate}
\item If $Z \in \operatorname{dom}(\riskutility)$, then $\partial\riskutility(Z)\neq\emptyset$.  
\end{enumerate}
\end{ass}
\begin{thrm}[Accuracy{/}Risk-Utility Trade-off Formulation]\label{thrm_Lagrangian_Formulation}
	Let $\rrrr$ be a risk-perspective and suppose that Assumption~\ref{ass_reg_cone} holds.  
	Let $\lambda\in (0,1)$, suppose that the set
		\begin{equation}
	\arginf{Z \in \Lp{\gggg}\cap \Phi }
	(1-\lambda)\ee{
		\left\|
		X-Z
		\right\|^2
	}
	+\lambda
	\riskutility(X-Z)
	,
	\label{thrm_existsnce_uniqueness_eq_defining_optimization_problem}
	\end{equation}
	is non-empty, and let $\eeRA{X}$ denote any member of that set.  
	Then~\eqref{PG} admits a solution, for the risk-aversion level $M_{\lambda}$ defined implicitly by
	\begin{equation}
M_{\lambda}\triangleq \frac{(1-\lambda)}{2\lambda}\cdot \riskutility\left(
X-
\eeRA{X}
\right)
\label{key}
.
	\end{equation}
\end{thrm}
\begin{remark}
The parameter $\lambda \in [0,1)$, controls the confidence in the investor's model.  If $\lambda\approx 0$, then the investor is highly certain that model is correct.  Conversely, if $\lambda\approx 1$ then they have low confidence in their choice of model, and therefore would be incorporating a high-level of risk-aversion in their estimates.  
\end{remark}
\begin{remark}\label{rem_IC_rep}
If it exists, specifying a map taking $X$ to a minimizer of $f\square g$, where
\begin{equation}
f(Y)\triangleq (1-\lambda)\ee{
	\left\|
	Y
	\right\|^2
}
+\lambda
\riskutility(Y)
;\qquad
g(Y)\triangleq \iota_{\Lp{\gggg}\cap\Phi}(Y).
\label{eq_inf_conv_descr}
\end{equation}
is equivalent to specifying a selection, taking $X$ to an element of the set~\eqref{thrm_existsnce_uniqueness_eq_defining_optimization_problem}.  
\end{remark}
\begin{proof}[Proof of Theorem~\ref{thrm_Lagrangian_Formulation}]
	Since taking the $\operatorname{arginf}$ is invariant under multiplication by a positive constant and addition of real numbers, then~\eqref{thrm_existsnce_uniqueness_eq_defining_optimization_problem} may be re-parameterized as
	\begin{align}
	&
	\arginf{Z \in \Lp{\ffff[]} } 
	\frac1{2}\ee{
		\left\|
		X-Z
		\right\|^2
	}
	+\tilde{\lambda}
	\left[\riskutility(X-Z)- M_{\lambda}\right]
	+\iota_{\Phi\cap \Lp{\gggg}}(Z)
	;
	\label{thrm_existsnce_uniqueness_eq_defining_optimization_problem_alt}
	\end{align}
	note that the finiteness of $M_{\lambda}$ is guaranteed by Definition~\ref{defn_perspective}(iii)%
	.

	Since $\Lp{\gggg}$ is linear, it is convex; hence $\Lp{\gggg}\cap \Phi$ is convex.  By assumption $\Phi$ is closed and since $\Lp{\gggg}$ is closed then $\Lp{\gggg}\cap\Phi$ is also closed.  
	
	This together with Assumption~\ref{ass_reg_cone} implies that the \citep[SCQ1 condition]{jeyakumar1992generalizations} holds.  Since there is no additional affine constraint in the optimization problem of Equation~\ref{PG}, then the \citep[A Generalized Slater's Condition; Corollary 5.1]{jeyakumar1992generalizations} %
	holds.  Thus $\eeRA{X}$ is optimal for the problem of Equation~\ref{PG} if and only if
	\begin{equation}
	\begin{aligned}
	0\in &\partial \frac1{2}\ee{
		\left\|
		X-Z
		\right\|^2
	}
	+ \tilde{\lambda} \partial\riskutility(X-Z)
	+ N_{\Phi\cap \Lp{\gggg}}(Z)
	\\
	= &
	\partial \left(\frac1{2}\ee{
		\left\|
		X-Z
		\right\|^2
	}
	+ \tilde{\lambda} 
	\riskutility(X-Z)-M_{\lambda}
	+ \iota_{\Phi\cap \Lp{\gggg}}(Z)
	\right) 
	.
	\end{aligned}
	\label{eq_KKT_gen_slater}
	\end{equation}
	and $\tilde{\lambda} 
	\riskutility(X-Z)-M_{\lambda}
	=0$.  
	
	It follows from \citep[Fermat's rule; Theorem 16.3]{ConvexMonoCMS} that $\eeRA{X}$ indeed satisfies Equation~\eqref{eq_KKT_gen_slater}, and the definition of $M_{\lambda}$ implies that 
	$$
	\tilde{\lambda}
	\riskutility(X-\eeRA{X})-M_{\lambda}
	=0
	.
	$$  
	Hence the two problems are equivalent, whenever~\eqref{thrm_existsnce_uniqueness_eq_defining_optimization_problem} admits a solution.  
\end{proof}
\begin{ex}[Additional Risk Constraint on Estimator]\label{ex_Lagrangian_mRC}
	Consider the feature set $\Phi=\aaaa_{R}^{g,m}$ of Example~\eqref{ex_accept_RM}.  Under the assumptions of Theorem~\ref{thrm_Lagrangian_Formulation}, if it exists, $\rrrr(X)$ is of the form
	\begin{equation}
	\begin{aligned}
	\rrrr(X) \in\arginf{Z \in \Lp{\gggg}}&
	(1-\lambda)\ee{
		\left\|
		X-Z
		\right\|^2
	}
	+\lambda
	\riskutility(X-Z)\\
	\mbox{subject to: }& R(g(Z))\leq m
	.
	\end{aligned}
	\label{eq_super_powers}
	\end{equation}
	Suppose that there exits $\eta>0$ such that
	\begin{equation}
	\begin{aligned}
	\arginf{Z \in \Lp{\gggg}\cap \aaaa_{R}^{g,m}}&
	(1-\lambda)\ee{
		\left\|
		X-Z
		\right\|^2
	}
	+\lambda
	\riskutility(X-Z)
	\\
	=
	\arginf{Z \in \Lp{\gggg}}&
	(1-\lambda)\ee{
		\left\|
		X-Z
		\right\|^2
	}
	+\lambda
	\riskutility(X-Z) +\eta \left(R\left(g(Z)\right)-m\right)\\
	=
	\arginf{Z \in \Lp{\gggg}}&
	(1-\lambda)\ee{
		\left\|
		X-Z
		\right\|^2
	}
	+\lambda
	\riskutility(X-Z) +\eta R\left(g(Z)\right)
	.
	\end{aligned}
	\label{eq_super_powers_alt_Lagrangian}
	\end{equation}
	
	In the special case where $\lambda=0$ (or equivalently where $M=\infty$), the $\rrrr$-conditioning operator reduces to a proximal mapping operator
	$$
	\rrrr^0(\cdot)=\Prox{\eta R\circ g}(\cdot).
	$$
\end{ex}
The result of Theorem~\eqref{thrm_Lagrangian_Formulation} was contingent on~\eqref{thrm_existsnce_uniqueness_eq_defining_optimization_problem} admitting a solution.  The next result confirms that this is always the case when $\lambda \in [0,1)$.  Moreover, that the solution is unique and arises from a well-posed problem.  
\begin{thrm}[%
	Well-Posedness]\label{thrm_existsnce_uniqueness}
	Let $X$ be in $\Lp{\ffff}$ %
	 and let $\rrrr$ be a risk-perspective, and let $\lambda \in [0,1)$.  Then 
	\begin{enumerate}[(i)]
\item (Existence and Uniqueness:) The set~\eqref{thrm_existsnce_uniqueness_eq_defining_optimization_problem} contains exactly one element,
\item (Continuity wrt $X$:) 
\begin{enumerate}[(a)]
\item If $\rho$ is finite-valued and $\utility$ is continuous, then the map 
\begin{equation}
X\mapsto \inf_{Z \in \Lp{\gggg}\cap \Phi }
(1-\lambda)\ee{
	\left\|
	X-Z
	\right\|^2
}
+\lambda
\riskutility(X-Z)
\label{eq_optim_val_defn}
,
\end{equation}
is continuous.  Moreover, if $\riskutility$ is Lipschitz, then~\eqref{eq_optim_val_defn} is Lipschitz.  
\item Furthermore, if the map $\riskutility$ is Fr\'{e}chet Differentiable on $\Lp{\ffff[]}$ and the map
$$
X \mapsto (\nabla \riskutility)\left(X - \rrrr(X)\right),
$$
is $k$-Lipschitz, where $k\in (0,\infty)$, then the map $X\mapsto \rrrr\left(X\right)$ is Lipschitz with constant $\left(1+\frac{\lambda k}{2(1-\lambda)}\right)$.  
\end{enumerate}
	\end{enumerate}
\end{thrm}

\begin{proof}[Proof of Theorem~\ref{thrm_existsnce_uniqueness}]
	Recall that, since $\hhh$ is Hilbert, then so is $\Lp{\ffff[]}$.  Furthermore, note that $\Lp{\gggg}$ is a Hilbert subspace of $\Lp{\ffff}$.  Moreover, note that the $\operatorname{arginf}$ operation is invariant under multiplication by positive constants and addition by constants; whence~\eqref{thrm_existsnce_uniqueness_eq_defining_optimization_problem} may be equivalently expressed as
	\begin{equation}
\arginf{Z \in \Lp{\ffff[]} } 
\frac1{2}\ee{
	\left\|
	X-Z
	\right\|^2
}
+\tilde{\lambda}
\left[\riskutility(X-Z)- M_{\lambda}\right]
+\iota_{\Phi\cap \Lp{\gggg}}(Z)
\label{thrm_existsnce_uniqueness_eq_defining_optimization_problem_altt}
.
	\end{equation} 

Since $\Lp{\gggg} $ is a closed linear subspace of $\Lp{\ffff}$, it is a closed convex subset of $\Lp{\ffff}$.  Therefore 
$\iota_{\Lp{\gggg}\cap\Phi}$ is convex and strongly-l.s.c.\ on $\Lp{\ffff}$.%

Since the set $\aaarho$ is non-empty and closed, \citep[Lemma 2.3]{filipovic2007convex} implies that $\rho$ is strongly-l.s.c.  
Define the functions $f$ and $g$ by
$$
\begin{aligned}
f(X)\triangleq & \frac1{2}\ee{\|X\|^2} + \left[\riskutility(X)-M_{\lambda}\right],\\
g(X)\triangleq & \iota_{\Lp{\gggg}[]\cap \Phi}(X),\\
h(X)\triangleq & f\square g (X),
\end{aligned}
$$
where $\square$ is the infimal convolution operator.  Notice that $h$ is minimized by the map taking $X$ to an element (if it exists) of the set~\eqref{thrm_existsnce_uniqueness_eq_defining_optimization_problem_altt}.  

Since $\riskutility$ is bounded below by $M\triangleq \inf_{Z \in \Lp{\ffff[]}}\riskutility(Z)$ and this quantity is finite, by definition of $\utility$, then
$$
\lim\limits_{\sqrt{\ee{\|X\|^2}}\mapsto \infty} \frac{f(X)}{\sqrt{\ee{\|X\|^2}}} \geq \frac1{2}\sqrt{\ee{\|X\|^2}} + \frac{M}{\sqrt{\ee{\|X\|^2}}} = \infty,
$$
therefore $f$ is super-coercive.  Moreover, since $\riskutility$ is convex, strongly lsc, and $\frac1{2}\ee{\|X\|^2}$ is strictly convex and strongly continuous, then $f$ is super-coercive, strictly convex, and strongly lsc.  Therefore, \citep[Proposition 12.14 (i)]{ConvexMonoCMS} implies that $f\boxdot g=h$; that is the infimal convolution $f\square g$ is exact.  This gives the existence of an element of~\eqref{thrm_existsnce_uniqueness_eq_defining_optimization_problem_altt}.  

The assumption that
$$
\emptyset\neq 
\Lp{\gggg[]}\cap \Phi\cap \aaaa_{\rho} \cap \operatorname{dom}\left(\riskutility\right)
=
\operatorname{dom}(f)\cap \operatorname{dom}(g)
$$
implies that \citep[Coroally 11.16 (i)]{ConvexMonoCMS} is applicable; whence $f+g$ is coercive and \textit{uniquely} attains its minimizers at every single point in $\Lp{\ffff[]}$.  In particular, this implies that $h(X)$ is minimized for every $X \in \Lp{\ffff[]}$.  This gives the uniqueness of an element of~\eqref{thrm_existsnce_uniqueness_eq_defining_optimization_problem_altt}.  Hence, (i)-holds.  

If $\rho$ is finite-valued, then \citep[Corollary 2.3]{KainaCRMs} implies that $\rho$ is continuous on $\Lp{\ffff[]}$.  Therefore, $f$ is continuous on $\Lp{\ffff[]}$, convex, and bounded below by $M$.  Since $g$ is also convex and bounded below (by $0$), then the results of \citep[pp. 37-38]{moreau1966fonctionnelles}, implies that $h$ is itself continuous.  Again appealing to the equivalence between problems~\eqref{thrm_existsnce_uniqueness_eq_defining_optimization_problem_altt} and~\eqref{thrm_existsnce_uniqueness_eq_defining_optimization_problem} yields (ii)-(a).  

Let $f,g$ and $h$ be as in~\eqref{eq_inf_conv_descr}.  Suppose that $\riskutility$ is Fr\'{e}chet differentiable on $\Lp{\ffff[]}$ and and suppose that the map $x\mapsto (\nabla \riskutility)(X-\rrrr(X))$ is $k$-Lipschitz.  Since the infimal convolution $f\square g$ is exact, and $\iota_{\Lp{\gggg}\cap \Phi}$ is lsc, proper, and convex, then \cite[Theorem 1.a]{stromberg1994study} implies that
\begin{equation}
\nabla \left(f\square g\right) = \nabla g = \lambda \left(\nabla \riskutility\right) + \left(\nabla (1-\lambda)\ee{\|\cdot\|}^2\right)
\label{eq_strombergwinninghappytimes_1}
.
\end{equation}
Since $\rrrr(X)\in \Lp{\gggg}\cap \Phi$ then by Fermat's rule \cite[Theorem 16.3]{ConvexMonoCMS}, 
\begin{equation}
\begin{aligned}
\rrrr\left(X\right)\triangleq \arginf{Z \in \Lp{\ffff[]}}[1] g(Z)+ f(X-Z)=& 
\arginf{Z \in \Lp{\gggg}\cap \Phi}[1] f(X-Z)\\
=&\left\{
Y \in \Lp{\ffff[]} : 0 \in \partial f(X-Y)
\right\}\\
=&\left\{
 Y \in \Lp{\ffff[]} : 0 \in \nabla f(X-Y)
\right\}.
\end{aligned}
\label{eq_fermatification_1}
\end{equation}

Combining~\eqref{eq_strombergwinninghappytimes_1} and~\eqref{eq_fermatification_1} we infer that
$$
\begin{aligned}
0=&\nabla h(X) \\
= &\lambda \nabla\riskutility (X-\rrrr(X)) + (1-\lambda)\nabla\ee{\left\|X-\rrrr(X)\right\|^2}\\
=& \lambda \nabla\riskutility (X-\rrrr(X)) + (1-\lambda)2(X-\rrrr(X))
\\
\therefore R(X) =& \frac{\lambda}{2(1-\lambda)} (\nabla\riskutility) (X-\rrrr(X)) +X
.
\end{aligned}
$$

Whence, the triangle inequality and the assumption that $X \mapsto (\nabla\riskutility)(X-\rrrr(X))$ is $k$-Lipschitz imply that
$$
\begin{aligned}
\sqrt{\ee{\|\rrrr(X)-\rrrr(Y)\|^2}} \leq & \sqrt{\ee{\|X- Y\|^2}} 
\\
+& \frac{\lambda}{2(1-\lambda)}\sqrt{
	\ee{
		\left\|
		(\nabla\riskutility)(X-\rrrr(X)) - 
		(\nabla\riskutility)(Y-\rrrr(Y))
		\right\|
	}
}\\
\leq &
\sqrt{\ee{\|X- Y\|^2}} + \frac{\lambda}{2(1-\lambda)} k \sqrt{\ee{\|X- Y\|^2}}
.
\end{aligned}
$$
Therefore (ii)-(b) holds.  
\end{proof}
\begin{remark}[Some Implicit Notation]\label{rem_not}
	Let $\rrrr=(\gggg,\rho,\utility,\Phi,M)$ be a risk-perspective.  From now on, we will denote by $\rrrr^{\lambda}$ the risk-perspective $(\gggg,\rho,\utility,\Phi,M_{\lambda})$ where $M_{\lambda}$ is defined as in~\eqref{key}.  
\end{remark}
Theorems~\ref{thrm_Lagrangian_Formulation} and~\ref{thrm_existsnce_uniqueness} establish the well-posedness of the solution operator of problem~\eqref{PG}.  %
We may therefore make the following definition.  
	\begin{defn}[$\rrrr$-Conditioning]\label{defn_RA_Expectation}
	Let $\rrrr$ be a risk-perspective satisfying Assumption~\ref{ass_reg_cone}, the map defined by
\begin{equation}
	\begin{aligned}
		\operatorname{dom}\left(\utility\right)
		 &\rightarrow \Lp{\gggg}\cap \Phi\\
		X & \mapsto \arginf{Z \in \Lp{\gggg}\cap \Phi}
		(1-\lambda)\ee{
			\|X-Z\|^2
		}
		+\lambda
		\riskutility(X-Z)%
		,
	\end{aligned}
\label{defn_RA_Expectation_eq_defining}
\end{equation}
is called the $\rrrr$-conditioning onto $\gggg$, and will be denoted by $\eeRA{X}$ (or by $\rrrr(X)$ whenever the context is clear).  
	\end{defn}
If we wish to make explicit reference to the parameter $\lambda$ or $\gggg$ defining the risk-perspective $\rrrr$, we will denote the $\rrrr$-conditioning of $X$ by $\eeRA{X}[\gggg]$.  
\begin{remark}[Alternative Parameterization]\label{rem_reparemterization_of_lambda}
Since the $\operatorname{argmin}$ multifunction is invariant with respect to multiplication of a positive scalar, then multiplying~\eqref{defn_RA_Expectation_eq_defining} by $\frac{2}{1-\lambda}$ implies that
\begin{equation}
\begin{aligned}
\rrrr(X)=& \arginf{Z \in \Lp{\gggg}\cap \Phi}
\frac1{2}\ee{
	\|X-Z\|^2
}
+\tilde{\lambda}
\riskutility(X-Z)%
,
\end{aligned}
\label{defn_RA_Expectation_eq_defining_alt}
\end{equation}
where $\tilde{\lambda}\triangleq \frac{2\lambda}{(1-\lambda)}$.  

Note, further that if $\rrrr(X)$ is given in the form~\eqref{defn_RA_Expectation_eq_defining_alt}, then $\lambda$ may be recovered and found to be $\lambda =\frac{\tilde{\lambda}}{\tilde{\lambda}+2}.$  
\end{remark}
\subsection{Risk-Averse Valuation}\label{s_RA_Pricing}
Using Theorems~\ref{thrm_Lagrangian_Formulation} and~\ref{thrm_existsnce_uniqueness}, the risk-averse valuation Problem~\ref{RA} is well-defined, and uniquely characterized by $\rrrr$.%

\begin{defn}[Risk-Averse Value]\label{defn_ra_val}
	Let $\rrrr$ be a risk-perspective with $\gggg=\ffff[0]$, $f:\hhh\rightarrow \rr$ a Borel-measurable payoff function, and assume that~\ref{ass_reg_cone} is satisfied.  	
	The $\rrrr$-risk-averse value, denoted $V^{\rrrr}$, of the derivative with payoff $f(X_T)$ is defined to be
	\begin{equation}
	\begin{aligned}
	V^{\rrrr}(T,\lambda)\triangleq &\eeRA{f(X_T)}%
	.
	\end{aligned}
	\end{equation}
\end{defn}
\begin{defn}[Mispricing Risk]\label{defn_RL_RAP}
	Let $Z \in \Lp{\gggg}$ and $\rrrr$ be a risk-perspective.  The quantity 
	$$
	\epsilon^{f}_T(Z)\triangleq \riskutility(f(X_T)-Z),
	$$
	will be called the $\rrrr$-mispricing risk.  
\end{defn}
The mispricing risk represents the risk of estimating the payoff $f(X_T)$ by the optimal estimator $\eeRA{X}$, as quantified by $\utility$.  The next result bounds the mispricing risk of $\eeRA{f(X_T)}$, above by the $\rrrr$-mispricing risk of $\ee{f(X_T)}[\ffff[0]][\qq]$.  
\begin{thrm}[Mispricing Risk Reduction]\label{thrm_risk_reduction}
	Fix a risk-perspective $\rrrr$ such that $\Phi=\Lp{\ffff[]}$.  
	\begin{enumerate}[(i)]
		\item For every $\lambda \in [0,1)$, the mispricing risk-aversion levels $M_{\lambda}$ is bounded by%
		\begin{equation*}
		-\infty 
		<
		\inf_{Z \in \Lp{\gggg}%
		} 
		\epsilon^{f}_T(Z)
		\leq 
		\epsilon^{f}_T\left(
		V^{\rrrr}(T,\lambda)
		\right)
		\leq
		\epsilon^{f}_T(\ee{f(X_T)}[\gggg][\qq])
		,
		\end{equation*}
		\item The map $\lambda \mapsto \epsilon^{f}_T\left(
		V^{\rrrr}(T,\lambda)
		\right)$ is monotonically decreasing,
		\item The mispricing risk-aversion level, asymptotically achieves minimal mispricing risk, in the sense that
		\begin{equation}
		\lim\limits_{\lambda \uparrow \infty} 
		\epsilon^{f}_T\left(
		V^{\rrrr}(T,\lambda)
		\right)
		\downarrow
		\inf_{Z \in \Lp{\gggg}\cap \Phi} 
		\epsilon^{f}_T\left(
		Z
		\right)
		\label{eq_consequence_main_thrm_2}
		.
		\end{equation}
	\end{enumerate}
\end{thrm}
The proof of Theorem~\ref{thrm_risk_reduction} will rely on the following technical Lemma.  
\begin{lem}%
	\label{lem_Asym_prop_RACE}
	The operator $\eeRA{X}$ has the following properties
	\begin{enumerate}[(i)]
		\item For any $X \in \Lp{\gggg}\cap \Phi$, 
		$$
		\lim\limits_{\lambda \uparrow \infty} \riskutility\left(\eeRA{X}\right) \downarrow \inf_{Z \in \Lp{\gggg}\cap \Phi} \riskutility(Z),
		$$
		\item The map $\lambda\mapsto  \riskutility\left(\eeRA{X}\right)$ is monotonically decreasing.
	\end{enumerate}
\end{lem}
\begin{proof}[Proof of Lemma~\ref{lem_Asym_prop_RACE}]
	Consider the proper, lsc, convex 
	$$
	f^X(\cdot)\triangleq 	\iota_{\Lp{\gggg}\cap\Phi}(\cdot) + \frac{2\lambda}{1-\lambda} \riskutility(X-\cdot).
	$$
	Then for any $X \in \Lp{\ffff[]}$, the functions $\eeRA{\cdot}$ and $\operatorname{Prox}_{\tilde{\lambda}f^X}(\cdot)$ agree at $X$; that is
	$$
	\eeRA{X}=\operatorname{Prox}_{\tilde{\lambda}f^X}(X)
	. 
	$$ 
	Note that $X \in \operatorname{dom}(\riskutility+\iota_{\Lp{\gggg}\cap\Phi})$ if and only if $X \in \operatorname{dom}(\riskutility)$ and $\iota_{\Lp{\gggg}\cap\Phi}(X)<\infty$; that is $X \in \operatorname{dom}(\riskutility+\iota_{\Lp{\gggg}\cap\Phi})$ if and only if $X$ is $\gggg$-measurable and in $\Phi$.  
	Therefore, the result follow from \citep[Proposition 12.33]{ConvexMonoCMS}.  
\end{proof}
\begin{proof}[Proof of Theorem~\ref{thrm_risk_reduction}]
	Consider the map
	$$
	\left[
	\lambda\mapsto 
	\epsilon^{f}_T\left(
	V^{\rrrr}(T,\lambda)
	\right)
	\right] \triangleq 
	\left[
	\lambda\mapsto 
	\riskutility\left(f(X_T)-\eeRA{f(X_T)}\right)
	\right]
	.
	$$
	The result, then follows from Lemma~\ref{lem_Asym_prop_RACE} and the monotonicity of the map $\lambda\mapsto \tilde{\lambda}$, of~\eqref{key}.  
\end{proof}
\section{Asymptotic Solution to{~\ref{intro_eq_RF}}}\label{s_sol_Uncertainty}

We now return to problem~\ref{intro_eq_RF}.  We will provide a solution to the problem, under the assumption that $\mathcal{X}=\Lp{\ffff[]}$ and that the penalty term is equivalent to the minimal penalty of a proper, lsc, convex risk-measure on $\Lp{\ffff[]}$.  That is, we are interested in solving for the optimizer of the following penalized sub-linear expectation problem%
\begin{equation}
\arginf{Z \in \Lp{\ffff[]}[D]} \sup_{Q \in \qqqq}\eemes{\|f(X_T)-Z\|^2}[Q] %
- \rho^{\star}\left(
\frac{dQ}{d\pp}
\right)
\label{eq_PSLE}
\tag{PSLE}
.
\end{equation}  

Our approach relies on the $\rrrr$-condition operators, as well as the theory of $\Gamma$-convergence.  To this end, we take a moment to review the latter, from a geometric perspective.  
\subsection{$\Gamma$-Convergence}\label{ss_Gamma_Conv}
Let $\rrbar$ denote the set $[-\infty,\infty]$ and let $(X,d)$ be a complete metric space.  

Many mathematical objects are described by the optimization of a specific loss-function $\ell$.  However, it is not uncommon for this loss-function to be challenging to compute either analytically or numerically.  Many authors, such as \cite{gonon2019linearized}, overcome this issue by approximating the problem by a more tractable problem and arguing that the solutions to the problems are also close.

Pioneered in \cite{ADeGiorgiOriginal}, the theory of $\Gamma$-convergence describes the precise conditions required for the optimizers of a sequence of loss-functions $\{\ell_n\}_{n\in\nn}$ to converge to the optimizer of the loss-function $\ell$.  %
Formally, $\Gamma$-convergence is a tool describing precisely when the following problem
\begin{equation}
\lim\limits_{n \uparrow \infty}\arginf{x \in X} \ell_n(x)
\in
\arginf{x \in X} \ell(x)
,
\label{eq_gamma_prob}
\tag{{$\Gamma$}}
\end{equation}
admits a solution.  The solution of Problem~\eqref{eq_gamma_prob} given in \cite{ADeGiorgiOriginal}, stems from the following result of Weierstrass (note the connection with Definition~\ref{defn_perspective}(iii)).
\begin{thrm}[Weierstrass' Theorem;{ \citep[Theorem 2.2]{GammaMEgapackSoNiceandConvenient}}]\label{thrm_weistrass_Theorem}
	Let $(X,d)$ be a complete metric space and $\ell:X\rightarrow \rr$ be a lower semi-continuous function which is mildly coercive, that is there exists a sequentially compact subset $K$ of $X$ such that
	\begin{equation}
	\inf_{x \in K} \ell(x)
	= 
	\inf_{x \in X} \ell(x)
	\label{eq_defn_coercive}
	.
	\end{equation}
	Then $\ell$ admits a minimizer on $X$ if, in addition
	\begin{equation}
	-\infty <\inf_{x \in X}\ell(x)
	\label{defn_weistrass_property}
	.
	\end{equation}
\end{thrm}

We geometrically unpack Theorem~\eqref{thrm_weistrass_Theorem}.  Let us first beginning with the definition of an epigraph, which, in the words of \cite{timsully}, can "be interpreted as the set of points lying on or above [a function's graph]."  
\begin{defn}[Epigraph]\label{defn_epi}
	The epigraph of a function $\ell:X\rightarrow \rrbar$, is the subset of $X \times \rrbar$ defined by
	$$
	\epi{\ell}\triangleq \left\{
	(x,y) \in X \times \rrbar : y\geq F(x)
	\right\}
	.
	$$
\end{defn}

Minimization is understood as travelling along the epigraph until the lowest point is hit.  The assumption of lower semi-continuity (lsc-ity), in Theorem~\ref{thrm_weistrass_Theorem}, states that if a solution exists, then the procedure of moving along the epigraph does, in fact, arrive at it.  In other words, the epigraph is closed.
\begin{lem}[Epigraphical Implications of lsc-ity;{\cite{timsully}}]\label{lem_epi}
	A function $\ell:X \rightarrow \rrbar$ is lsc if and only if $\epi{\ell}$ closed in $X \times \rrbar$.  Moreover, if $X$ is compact and $\ell$ is lsc, then there exists $\hat{x}\in X$ satisfying minimizing $\ell$.  
\end{lem}

Geometrically speaking, Lemma~\ref{lem_epi} hints at the fact that if the minimizers of $\ell_n$ are to converge to a minimizer of $\ell$, then according their epigraphs must converge.  In order to make this convergence meaningful, we turn to the Pompeiu-Hausdorff space over $X \times \rrbar$, denoted by $\PH{(X,d)}$.  The points in this metric space are closed subsets of $X\times \rrbar$ and the space itself is topologized through the metric $d_{\PH}$.  The metric $d_{\PH}$ is based on the familiar formula for the distance between a closed subset $B$ of $(X,d)$ and a point therein
\begin{equation}
d_{\PH}(\{x\},B)\triangleq
\inf_{b \in B} d(a,B)
\label{eq_defn_PH_dist_sngl}
;
\end{equation}
more generally, for any two closed subsets of $(X,d)$, their distance is defined by
\begin{equation}
d_{\PH}(A,B)\triangleq
\max\{\,\sup_{a \in A} \inf_{b \in B} d(a,b),\, \sup_{b \in B} \inf_{a \in A} d(a,b)\,\}
\label{eq_defn_PH_dist}
.
\end{equation}
Lemma~\ref{lem_epi} can therefore be seen as a way of topologizing the set of lsc functions on $X$, with a rather tame topological space.

If $(X,d)$ is compact and $\{\ell\}_{n \in\nn}$ forms a regular enough sequence, it turns out that the point-wise convergence, in $\PH{(X,d)}$, of $\{\epi{\ell_n}\}_{n \in \nn}$ to $\epi{\ell}$ is exactly what is needed for \eqref{eq_gamma_prob} to hold.  However, in general, point-wise convergence is too stringent to provide a characterization of the problem.  $\Gamma$-convergence, is thus a weakened version of point-wise convergence in $\PH{(X,d)}$, but remains equivalent to point-wise if $(X,d)$ is compact.  

Our description of it, will exploit the Kuratowski limit, see for more details \cite{kuratowski2014topology}.  Let $\{A_n\}_{n\in\nn}$ be a sequence of subsets of $X$.  If both the following sets exist, 
\begin{equation}
\begin{aligned}
& \left\{
x \in X : \limsup\limits_{n \mapsto \infty} d(x,A_n) =0
\right\},\\
& \left\{
x \in X : \liminf\limits_{n \mapsto \infty} d(x,A_n) =0
\right\},\\
\end{aligned}\label{eq_Klim}
,
\end{equation}
are well-defined and %
are equal%
, then the Kuratowski limit, denoted by $
\Klim A_n$ is defined to be that limiting set.  
\begin{remark}[Non-LSC Functions]\label{rem_not_clos}
	Note that, unlike pointwise convergence in $\PH{(X,d)}$, Kuratowski convergence does not actually require the sets $\left\{A_n\right\}_{n \in \nn}$ to be in $\PH{(X,d)}$.  Therefore, if $\{\ell_n\}_{n \in \nn}$ is a sequence of functions, which are not lsc on $(X,d)$, then the Kuratowski limit of their epigraphs can still be meaningfully described, while the pointwise limit in $\PH{(X,d)}$ is meaningless.  
\end{remark}

\begin{ex}[Constant Epigraphical Limits]\label{ex_LSC_relax}
	Let $A$ be a subset of $X$ and consider the constant sequence $\{A_n\}_{n \in \nn}$, where $A_n\triangleq A$ for every $n \in \nn$.  In \citep[Example 4.12]{dal2012introduction} it is argued that $\Klim A_n$ is the closure of $A$ in $X$.  
\end{ex}

\begin{defn}[$\Gamma$-convergence]\label{defn_Gamma_conv_epi}
	A sequence of functions $\{\ell_n\}_{n \in \nn}$ from $(X,d)$ to $\rrbar$, is said to $\Gamma$-converge to a function $\ell:X\rightarrow \rrbar$ if $\Klim\epi{\ell_n}$ and
	$$
	\Klim\epi{\ell_n} = \epi{\ell}
	.
	$$
	In that case, $\ell$ is said to be the $\Gamma$-limit of $\{\ell_n\}_{n \in \nn}$ and is denoted by $\Glim \ell_n$.  
\end{defn}

\begin{ex}[Lsc Relaxation and Constant $\Gamma$-Limits]\label{ex_lsc_relax}
	Building on Example~\ref{ex_LSC_relax}, if $\ell$ is a function from a complete metric space $(\tilde{X},d)$ to $\rrbar$, $A=\epi{\ell}$, and $X=\tilde{X}\times \rrbar$, then $\Klim\epi{\ell}$ is the smallest closed set in $\tilde{X}\times \rrbar$ containing $\epi{\ell}$.  Therefore, if $\ell$ is lsc, then Lemma~\ref{lem_epi} implies that $\Klim\epi{\ell}=\epi{\ell}$.  
	
	However, in general $\ell$ need not be lsc.  In this case, it is a direct consequence of \citep[Remark 4.5 and Example 4.12]{dal2012introduction} that
	\begin{equation}
	\Klim\epi{\ell}=\epi{\ell^{lsc}},
	\label{eq_gamma_limt_const_expl}
	\end{equation}
	where $\ell^{lsc}$ denotes the largest lsc function dominated by $\ell$, aptly called the lsc-relaxation of $\ell$.  Therefore,~\eqref{eq_gamma_limt_const_expl} implies that
	\begin{equation}
	\Glim \ell = \ell^{lsc}
	\label{eq_gamma_limt_const_expl_Gamma_version}
	.
	\end{equation}
\end{ex}

The following is an example of when the $\Gamma$-limit of a sequence of functions coincides with its point-wise limit.  
\begin{prop}[{\cite[Proposition 5.4]{dal2012introduction};\citep[Remark 1.40]{braides2002gamma}(ii)}]\label{prop_mono}
	Let $\{\ell_n\}_{n \in \nn}$ be a point-wise monotonically increasing sequence of functions from $(X,d)$ to $\rrbar$, which are uniformly bounded below.  Then the $\Gamma$-limit exists and can be computed via
	$$
	\Glim\ell_n = \lim\limits_{n \uparrow \infty} \ell_n^{lsc} = \left(\lim\limits_{n \uparrow \infty} \ell_n\right)^{lsc}
	.
	$$
\end{prop}

\begin{thrm}[Properties of {$\Gamma$-}convergence; {\citep[Theorem 2.8]{focardi2012gamma}}]\label{prop_prop_of_Gamma}
	Let $\{\ell_n\}_{n \in \nn}$ be a sequence of functions on $(X,d)$ and suppose that $\Glim \ell_n$ exists.  Then
	\begin{enumerate}[(i)]
		\item (Lower Semicontinuity): $\Glim \ell_n$ is lower semicontinuous on $X$,
		\item (Stability Under Continuous Perturbation): If $g:X \rightarrow \rr$ is continuous, then
		$$
		\Glim (\ell_n + g) = \left(\Glim \ell_n \right) +g
		,
		$$
		\item (Stability Under Relaxation): For every $n \in \nn$ let $\{\tilde{\ell}_n\}_{n \in \nn}$ be a sequence of functions from $X$ to $\rrbar$ satisfying
		$
		\ell_n^{lsc}\leq \tilde{\ell}_n\leq \ell_n.
		$
		Then 
		$$
		\Glim \tilde{\ell}_n = \Glim \ell_n
		.
		$$
	\end{enumerate}
\end{thrm}

The first of the two critical ingredients in Theorem~\ref{thrm_weistrass_Theorem} was the lsc-ity of $\ell$ and the second was its mild coerciveness.  Analogously to the definition of equi-continuity, in general, when working with a sequence of functions, to be able to apply the analogous machinery to Theorem~\ref{thrm_weistrass_Theorem} we require that there exists a non-empty compact subset of $K$ satisfying
\begin{equation}
\inf_{x \in X}\ell_n(x) = \inf_{x \in X}\ell(x)
;\; (\forall n \in \nn)
\label{eq_equi_cont}
.
\end{equation}
The property described by~\eqref{eq_equi_cont} is called \textit{mild equi-coerciveness}.  A stronger condition, that we will make use of is \textit{equi-coerciveness}, which states that for every $t>0$, there exists a compact subset $K_t$ of $(X,d)$ satisfying
$$
\bigcup_{n \in \nn}\{x \in X: \ell_n(x)\leq t\}\subseteq K_t
.
$$
The next result, is set in the case where $(X,d)$ is induced by the structure of a Banach space.  In that setting, it is shown that equi-coerciveness can be interpreted as a type of uniform growth condition.  
\begin{lem}[Equi-coerciveness in Normed Vector Space]\label{ex_coercive}
	Suppose that $X$ is a linear space and the metric $d$ is induced by a norm $\|\cdot\|$ on $X$, making $(X,\|\cdot\|)$ into a Banach space.  Then $\{\ell_n\}_{n \in \nn}$ is an equi-coercive family on $(X,d)$ if and only if there exist a coercive function $f:X\rightarrow \rrbar$ satisfying the growth condition
	\begin{enumerate}[(i)]
		\item $\inf_{ n \in \nn}\ell_{n}(x)\geq f(x)$, for every $x \in X$,
		\item $\lim\limits_{\|x\|\mapsto \infty}\frac{f(x)}{\|x\|}=\infty$.
	\end{enumerate}
\end{lem}
\begin{proof}
	In \citep[Proposition 7.7]{dal2012introduction}, it is seen that $\{\ell_{n}\}_{n \in \nn}$ is a equicoercive if and only if there exists a coercive function $f:X\rightarrow \rrbar$ dominated by each $\ell_n$.  The function $f$ is dominated by each $\ell_n$ if and only if
	\begin{equation}
	\inf_{ n \in \nn}\ell_n(x)\geq f(x)
	;\;
	(\forall x \in X)
	\label{eq_domi}
	.
	\end{equation}
	
	Moreover, in \citep[Proposition 2.18]{GammaMEgapackSoNiceandConvenient}, it is shown that $f$ is a coercive function on a normed vector space if and only if it satisfies the super-linearity condition
	\begin{equation*}
	\lim\limits_{\|x\|\mapsto \infty}\frac{f(x)}{\|x\|}=\infty
	.
	\end{equation*}
\end{proof}

We are now in place to describe a solution to problem~\eqref{eq_gamma_prob}.  Keeping Example~\ref{ex_lsc_relax} in mind, the next result directly parallels Theorem~\ref{thrm_weistrass_Theorem}.  
\begin{thrm}[{The Fundamental Theorem of {$\Gamma$}-Convergence; \citep[Theorem 2.10]{braides2002gamma},\citep[Theorem 2.1]{focardi2012gamma}}]\label{thrm_FTOG}
	If $\{\ell_n\}_{n \in \nn}$ is a mildly equi-coercive sequence of functions from $X$ to $\rrbar$ for which the $\Gamma$-limit exists in $X$%
	, then
	$$
	\lim\limits_{k\uparrow \infty} \inf_{x \in X}\ell_{k_n}(x) = \inf_{x \in X}\Glim\ell_n(x)
	.
	$$
	If moreover, $\{\ell_n\}_{n \in \nn}$ is equicoercive, then $\lim\limits_{n\uparrow \infty} \arginf{x \in X}\ell_n(x)$ exists in $X$ and
	$$
	\lim\limits_{k\uparrow \infty} \arginf{x \in X}\ell_{k_n}(x) \in \arginf{x \in X}\Glim\ell_n(x)
	.
	$$
\end{thrm}
\begin{remark}
	From \citep[Remark 1.20]{braides2002gamma}, it follows that $\{\ell_n\}_{n \in \nn}$ is a mildly-equicoercive family of functions, if it is a coercive family of functions.  Many formulations of the Fundamental Theorem of $\Gamma$-convergence appearing in the literature, assume pre-compactness instead of equicoerciveness.  As discussed in \citep[Remark 1.20]{BraideBegin}, equicoerciveness is a strictly stronger property than pre-compactness.  Therefore, the Theorem~\eqref{thrm_FTOG} implies many formulations of the result (for example, that in \cite{BraideBegin}).  
\end{remark}

For completeness, we mention that Definition~\ref{defn_Gamma_conv_epi} is equivalent to the original definition of $\Gamma$-convergence, which directly lifts an alternative characterization of metric lsc-ity of a function to a sequence of functions (see \citep[Section 2]{focardi2012gamma} for background to the classical definition).  
For self-containedness, we summarize the original definition in the following proposition.
\begin{prop}[Original Definition; {\citep[Chapter 4]{dal2012introduction}}]\label{defn_Gamma_conv}
	Let $\{\ell_n\}_{n\in \nn}$ be a sequence of $\rr\cup \{\infty\}$-valued functions on a complete metric space $(X,d)$.  A function $\ell$ is the $\Gamma$-limit of $\{\ell_n\}_{n \in \nn}$ if and only if both
	\begin{enumerate}[(i)]
		\item (Lower Bound Inequality:) $\ell^{lsc}(x)\leq \liminf\limits_{n \in \nn} \ell_n(x_n)$ for \textit{every} net $\{x_n\}_{n \in \nn}$ converging to $x$ in $(X,d)$,
		\item (Upper Bound Inequality:)  $\ell^{lsc}(x)\geq \liminf\limits_{n \in \nn} \ell_n(y_n)$ for \textit{some} net $\{y_n\}_{n \in \nn}$ converging to $x$ in $(X,d)$
	\end{enumerate}
	where $\ell^{lsc}$ is the lsc relaxation of $\ell$.  
\end{prop}

\subsection{Solution to Objective Problem{~\ref{intro_eq_RF}}}\label{s_RNR}
We will require that the risk-perspective $\rrrr$ is compatible with the tools offered by the theory of $\Gamma$-convergence.  This comparability comes in the form of the Weierstrass property, which is reminiscent of the conditions of Theorem~\ref{thrm_weistrass_Theorem}, on which $\Gamma$-convergence was built.  
\begin{property}[Weierstrass Property]\label{property_weirestrass}
	A proper lower semi-continuous functional $\ell:\Lp{\ffff}\rightarrow (-\infty,\infty]$ has the Weierstrass property if
	\begin{enumerate}[(i)]
		\item $\ell$ is coercive; that is $\ee{\|Z\|}\mapsto \infty$ implies that $\ell(Z)\mapsto \infty$,
		\item $-\infty<\inf_{Z \in \Lp{\ffff}} %
		\ell(Z)
		.
		$
	\end{enumerate}
\end{property}
\begin{defn}[Weierstrass Risk-Perspective]\label{defn_WRP}
	We will say that the risk-perspective $\rrrr$ is Weierstrass, if $\riskutility$ has the Weierstrass property.  
\end{defn}

Under these assumptions,  to show the following key property of $\rrrr$-conditioning, with respect to a Weierstrass risk-perspective.  
\begin{thrm}[Asymptotic Solution of Problem{~\eqref{eq_PSLE}}]\label{thrm_Approximation_Theorem}
	Suppose that $\rrrr$ is a Weierstrass risk-perspective and $\rho$ admits the robust representation~\eqref{eq_robust_representation_Lp_risk_measures}. Then
	\begin{enumerate}[(i)]
		\item \textbf{Existence:} Problem~\eqref{eq_PSLE} admits a solution
		\item \textbf{Asymptotic Solution:} Moreover, the solution to~\eqref{eq_PSLE} can be expressed as
		\begin{equation}
		\lim\limits_{
			\underset{n \geq 1}{n \uparrow \infty}
		}
		\eeRA{X}[\gggg][\frac{n}{2+n}]
		\in 
		\arginf{Z \in \Lp{\ffff[]}[D]} \sup_{Q \in \qqqq}\eemes{\|f(X_T)-Z\|^2}[Q] %
		- \rho^{\star}\left(\frac{dQ}{d\pp}\right)
		\label{eq_thrm_Approximation_Theorem_statement}
		,
		\end{equation}
		\item \textbf{Minimum Variance Solution:} For every $\tilde{Z}\in \arginf{Z \in \Lp{\gggg}\cap \Phi} \riskutility(X-Z)$ the following minimality property holds
		\begin{equation}
		\ee{\left\|
			X-
			\left(\lim\limits_{n \uparrow \infty;n\geq 1}\rrrr^{\frac{n}{2+n}}(X)\right)
			\right\|^2}
		\leq 
		\ee{\left\|
			X-\tilde{Z}
			\right\|^2}
		\label{eq_minimal_var_solution}
		.
		\end{equation}
	\end{enumerate}
\end{thrm}
The proof of Theorem~\ref{thrm_Approximation_Theorem} will be established in a series of steps, revolving around the use of Theorem~\ref{thrm_FTOG}.  
\begin{proof}%
    Fix $X \in \Lp{\ffff[]}$.  
\hfill\\\textbf{Existence of a solution}\hfill\\
Let $\utility\triangleq -\left\|\cdot \right\|^2$.  
Since $\Phi$ and $\Lp{\gggg}$ are closed, then \citep[Example 1.25]{ConvexMonoCMS} implies that $\iota_{\Lp{\gggg}\cap \Phi}$ is lsc.  Since the sum of lsc functions is lsc, and the sum of functions which are bounded from below is bounded from below, then $\iota_{\Lp{\gggg}\cap \Phi}(\cdot)+\riskutility(\cdot)$ is lsc and bounded from below by the minimum value of $\riskutility$ on $\Lp{\gggg}\cap \Phi$.  

For notational simplicity denote the set $\arginf{Z \in \Lp{\gggg}\cap \Phi}\riskutility(X-Z)$ by $\mathcal{A}$.  The Weierstrass property points (ii) implies $\riskutility$ is coercive.  Since
$$
\iota_{\Lp{\gggg}\cap \Phi}(\cdot)+\riskutility(X-\cdot)\geq \riskutility(    X-\cdot),
$$
then $\iota_{\Lp{\gggg}\cap \Phi}(\cdot)+\riskutility(X-\cdot)$ must also be coercive.  

Therefore, \citep[Remark 1.20]{braides2002gamma} implies that $\iota_{\Lp{\gggg}\cap \Phi}(\cdot)+\riskutility(X-\cdot)$ is (mildly-)coercive.  By \citep[Theorem 2.2]{focardi2012gamma}, the Weierstrass Property (iii) and the (mild-)coerciveness of $\iota_{\Lp{\gggg}\cap \Phi}(\cdot)+\riskutility(X-\cdot)$, it follows that $\mathcal{A}$ is non-empty.  Hence a solution to the right-hand side of~\eqref{eq_thrm_Approximation_Theorem_statement} exists.  Hence (i) holds.

\hfill\\\textbf{Computation of $\Gamma$-limit}\hfill\\
By the Weierstrass Property of $\riskutility$, Theorem~\ref{thrm_weistrass_Theorem} implies that the real number $M\triangleq \inf_{Z \in \Lp{\ffff}} \riskutility(X-Z)$ is well-defined.  Define the family of functions $\left\{\tilde{\ell}_{n}:\Lp{\ffff}\rightarrow (-\infty,\infty]\right\}_{n \in\nn}$ by
$$
\tilde{\ell}_{n}(Z)\triangleq (n+1)\cdot \left(\riskutility(X-Z)-M\right)
+ \iota_{\Lp{\gggg}\cap \Phi}(Z)
.
$$
Then $\{\tilde{\ell}_{n}\}_{n \in \nn}$ is a family of lsc functions to $\rrbar$, which are uniformly bounded below by $0$, minimized on $\mathcal{A}\neq\emptyset$, and point-wise monotonically increasing in $\lambda$.  Therefore, Proposition~\ref{prop_mono} implies that their $\Gamma$-limit exists and is equal to the lsc relaxation of their pointwise limit.  Its pointwise limit is the indicator function $\iota_{\mathcal{A}}$.

The lower semi-continuity of $\riskutility(X-\cdot)$ implies that its levels sets must be closed; in particular, this is the case for the following level set
$$
\begin{aligned}
\mathcal{A}=&
\left\{
Z \in \Lp{\gggg}\cap \Phi:
\riskutility(X-Z)
- M=0
\right\}\\
=
&\left\{
Z \in \Lp{\gggg}:
\riskutility(X-Z)
+\iota_{\Lp{\gggg}\cap \Phi}(Z)
\leq M
\right\}
,
\end{aligned}
$$
where the first equality was obtained by the definitions of $M$ and of $\mathcal{A}$.  Since $\mathcal{A}$ was seen to be non-empty and closed, \citep[Example 1.25]{ConvexMonoCMS} implies that $\iota_{\mathcal{A}}$ is lsc.  Therefore, it is equal to its own lsc relaxation; whence, 
\begin{equation}
\Glim\tilde{\ell}_n = \lim\limits_{n \uparrow \infty}\tilde{\ell}_n=\iota_{\mathcal{A}}
+
\iota_{\Lp{\gggg}\cap \Phi}
=\iota_{\mathcal{A}\cap\Lp{\gggg}\cap \Phi}\triangleq \ell-M.
\label{eq_Gammalimt_form}
\end{equation}

For every $n \in \nn$, define $\ell_n(\cdot) \triangleq \tilde{\ell}_n(\cdot)+\frac1{2}\ee{\|X-\cdot\|^2}$.  
Theorem~\ref{prop_prop_of_Gamma} (ii) and the continuity of the map $Z\mapsto\frac1{2} \ee{\|X-Z\|^2}$ imply that 
\begin{equation}
\Glim \ell_n(\cdot)=
\Glim \frac1{2}\ee{\|X-\cdot\|^2}+\tilde{\ell}_n(\cdot)= \frac1{2}\ee{\|X-\cdot\|^2}+%
\iota_{\mathcal{A}\cap\Lp{\gggg}\cap \Phi}(\cdot)
\label{eq_proof_thrm_Approximation_Theorem_gamma_limit}
.
\end{equation}

\hfill\\\textbf{Proof of %
	Equicoercivity}\hfill\\
Since the sequence $\{\tilde{\ell}_{n}\}_{n \in \nn}$ is point-wise monotonically increasing, %
then 
\begin{equation}
\ell_n(\cdot)=\frac1{2}\eemes{\|X-\cdot\|^2}+\tilde{\ell}_{n}(\cdot)\geq 
\frac1{2}\eemes{\|X-\cdot\|^2}
+
\ell(\cdot)
\geq
\ell(\cdot)
\geq \riskutility
(X-\cdot)
-M
,
\label{eq_equi_cor_setup}
\end{equation}
for every $n \in \nn$.  The Weierstrass Property (iii), Equation~\eqref{eq_equi_cor_setup}, and Lemma~\ref{ex_coercive} %
imply that 
$
\{\ell_n%
\}_{n \in \nn}
$ is an equicoercive family on $\Lp{\ffff}$.

\hfill\\ \textbf{$\Gamma$-convergence}\hfill\\
The sequence $\{\ell_{n}\}_{n \in \nn}$ meets all the requirements for the use of Theorem~\ref{thrm_FTOG}, therefore
\begin{equation}
\lim\limits_{n \uparrow \infty} \arginf{Z \in \Lp{\ffff}
} \ell_n(Z) \in
\arginf{Z \in \Lp{\ffff}%
}
\frac1{2}\ee{\|X-Z\|^2}+%
\iota_{\mathcal{A}\cap\Lp{\gggg}\cap \Phi}(Z)
\label{eq_proof_thrm_Approximation_Theorem_gamma_limit_initial_limit_pre}
.
\end{equation}
Since both sides of~\eqref{eq_proof_thrm_Approximation_Theorem_gamma_limit_initial_limit_pre} can only assume finite values on $\Phi\cap \Lp{\ffff[]}\cap \dom{\riskutility}$, which was assumed to be non-empty within Definition~\ref{defn_perspective}, then~\eqref{eq_proof_thrm_Approximation_Theorem_gamma_limit_initial_limit_pre} may be rewritten as
\begin{equation}
\lim\limits_{n \uparrow \infty} \arginf{Z \in 
	\Lp{\gggg}\cap \Phi
} \ell_n(Z) \in
\arginf{Z \in \Lp{\gggg}\cap \Phi
}
\frac1{2}\ee{\|X-Z\|^2}+%
\iota_{\mathcal{A}\cap\Lp{\gggg}\cap \Phi}(Z)
\label{eq_proof_thrm_Approximation_Theorem_gamma_limit_initial_limit}
.
\end{equation}

\hfill\\\textbf{Simplification of The Limiting Function}\hfill\\
We first show that the left-hand side of~\eqref{eq_proof_thrm_Approximation_Theorem_gamma_limit_initial_limit} is equal to $\eeRA{X}[\gggg][\frac{n}{2+n}]$, we implicitly appeal to the reparameterization defined in Remark~\ref{rem_reparemterization_of_lambda}.  

Since $\arginf{}$ is invariant under addition of a constant then, the expression of the left-hand side of~\eqref{eq_proof_thrm_Approximation_Theorem_gamma_limit_initial_limit}, simplifies to
\begin{equation}
\lim\limits_{\lambda \uparrow \infty}
\arginf{Z \in \Lp{\gggg}\cap \Phi} \ell_n(Z)
=
\lim\limits_{n \uparrow \infty;n\geq 1}
\arginf{Z \in \Lp{\gggg}\cap \Phi} \frac1{2}\ee{\|X-Z\|^2}+n\riskutility(X-Z) 
\label{eq_proof_thrm_Approximation_Theorem_gamma_limit_simplification_LHS}
.
\end{equation}
Therefore,~\eqref{eq_proof_thrm_Approximation_Theorem_gamma_limit_initial_limit} may be written as
\begin{equation}
\lim\limits_{n \uparrow \infty;n\geq 1}
\arginf{Z \in \Lp{\gggg}\cap \Phi} \frac1{2}\ee{\|X-Z\|^2}+n\riskutility(X-Z) 
\in
\arginf{Z \in \Lp{\gggg}\cap \Phi}
\frac1{2}\ee{\|X-Z\|^2}+\iota_{\mathcal{A}}(Z)
\label{eq_proof_thrm_Approximation_Theorem_gamma_limit_initial_limit_simplif}
.
\end{equation}
Theorem~\ref{thrm_existsnce_uniqueness} implies that $\rrrr^{\frac{n}{2+n}}(X)$ is the unique element in each of the sets $\arginf{Z \in \Lp{\gggg}\cap \Phi} 
\frac1{2}\ee{\|X-Z\|^2}
+n\riskutility(Z)$.  Therefore, the left-hand side of~\eqref{eq_proof_thrm_Approximation_Theorem_gamma_limit_initial_limit_simplif} simplifies to
\begin{equation}
\lim\limits_{n \uparrow \infty;n\geq 1}
\rrrr^{\frac{n}{2+n}}(X)
\in
\arginf{Z \in \Lp{\gggg}\cap \Phi}
\frac1{2}\ee{\|X-Z\|^2}+\iota_{\mathcal{A}}(Z)
\label{eq_proof_thrm_Approximation_Theorem_gamma_limit_initial_limit_simplif_real_nice_boyyy}
.
\end{equation}

We now show that the right-hand side of~\eqref{eq_proof_thrm_Approximation_Theorem_gamma_limit_initial_limit_simplif} is equal to the optimizer of~\eqref{eq_PSLE}.  First, let $Z^{\star}$ be a cluster point of the sequence $\rrrr^{\frac{n}{2+n}}(X)$ in $\Lp{\ffff[]}$.  

Together, the definition of the infimum, the continuity of the map $Z \mapsto X-Z$ on $\Lp{\ffff[]}$, Equation \eqref{eq_proof_thrm_Approximation_Theorem_gamma_limit_initial_limit}, and the lsc-ity of $\riskutility$ 
imply that
\begin{equation}
\begin{aligned}
\inf_{Z \in \Lp{\gggg}\cap \Phi}\riskutility\left(
X-
Z
\right)
\leq &
\riskutility\left(
X-
Z^{\star}
\right) \\
= &
\riskutility\left(
\lim\limits_{n \uparrow \infty} 
\left(
X-
\rrrr^{\frac{n}{2+n}}(X)
\right)
\right)\\
\leq &
\liminf\limits_{n \uparrow \infty}
\riskutility\left(
X-
\rrrr^{\frac{n}{2+n}}(X)
\right)
.
\end{aligned}
\label{eq_proof_thrm_Approximation_Theorem_gamma_limit_RHS_1}
\end{equation}
However, Lemma~\ref{lem_Asym_prop_RACE}(i) implies that
\begin{equation}
\lim\limits_{n \uparrow \infty}
\riskutility\left(
X-
\rrrr^{\frac{n}{2+n}}(X)
\right)
=
\liminf\limits_{n \uparrow \infty}
\riskutility\left(
X-
\rrrr^{\frac{n}{2+n}}(X)
\right)
= 
\inf_{Z \in \Lp{\gggg}\cap \Phi}\riskutility\left(
X-Z
\right)
\label{eq_help_thank_Lemma}
.
\end{equation}
Hence,
\begin{equation}
\lim\limits_{n \uparrow \infty;n\geq 1}
\rrrr^{\frac{n}{2+n}}(X)
\in \mathcal{A} 
= 
\arginf{Z \in \Lp{\gggg}\cap \Phi}\riskutility(X-Z)
.
\label{eq_leconq}
\end{equation}
Combining~\eqref{eq_proof_thrm_Approximation_Theorem_gamma_limit_initial_limit_simplif_real_nice_boyyy} and~\eqref{eq_leconq} yields
\begin{equation}
\begin{aligned}
\lim\limits_{n \uparrow \infty;n\geq 1}
\rrrr^{\frac{n}{2+n}}(X)
\in &
\arginf{Z \in \Lp{\gggg}\cap \Phi}\riskutility(X-Z)
\bigcap 
\arginf{Z \in \Lp{\gggg}\cap \Phi}
\frac1{2}\ee{\|X-Z\|^2}+\iota_{\mathcal{A}}(Z)\\
=
&
\arginf{Z \in \Lp{\gggg}\cap \Phi}\riskutility(X-Z)
\bigcap 
\arginf{Z \in \Lp{\gggg}\cap \Phi\cap \mathscr{A}}
\frac1{2}\ee{\|X-Z\|^2}
\end{aligned}
.
\label{eq_leconq_supreme}
\end{equation}
Hence (iii) holds.    
In particular,\eqref{eq_leconq_supreme} implies that
\begin{equation}
\lim\limits_{n \uparrow \infty;n\geq 1}
\rrrr^{\frac{n}{2+n}}(X)
\in
\arginf{Z \in \Lp{\gggg}\cap \Phi}
\riskutility\left(
X-
Z
\right)
\label{eq_proof_thrm_Approximation_Theorem_gamma_limit_initial_limit_simplif_real_nice_boyyy_nice_buddy_boyy}
.
\end{equation}

Applying the robust representation~\eqref{eq_robust_representation_Lp_risk_measures} to~\eqref{eq_proof_thrm_Approximation_Theorem_gamma_limit_initial_limit_simplif_real_nice_boyyy_nice_buddy_boyy} yields (ii)%
.  
\end{proof}

We now discuss a scheme for computing $\rrrr(X)$%
.%
\section{Computation of ${\riskutility}$-Conditioning}\label{s_compute}
This section shows how the forward-backwards splitting algorithm of \cite{combettes2005signal} can be used to compute the solutions to~\eqref{PG} and~\eqref{RA}.  When the feature set $\Phi$ is the entire space $\Lp{\ffff[]}$; the algorithm can be interpreted as a type of iteration between a gradient descent algorithm and a Monte-Carlo scheme.%

In general, when $\Phi$ need not be the entire space.  This restriction renders the Monte-Carlo step in the proximal splitting algorithm invalid.  Instead, a restricted version of the conditional expectation operator, introduced in the next section, is used in its place.  We show that this restricted conditional expectation can be explicitly computable via a \cite[Dykstra Splitting Algorithm]{boyle1986method}.  
\subsection{Relative Conditional Expectation}
The definition of the $\Phi$-Relative Conditional Expectation, is based on the following Lemma.  
\begin{lem}\label{lem_aaa_ce}
Let $\Phi$ be an feature set.  For every $X \in \Lp{\ffff[]}$, there exists a unique element in each of the sets
\begin{align}
\label{eq_lem_aaa_ce_0}
&\arginf{Z \in \Lp{\ffff[]}} \iota_{\Lp{\gggg}}(Z) + \iota_{\Phi}(Z) + \ee{\|X-Z\|^2}\\
&\arginf{Z \in \Lp{\ffff[]}} \iota_{\Phi}(Z) + \ee{\|X-Z\|^2}
\label{eq_lem_aaa_ce}
.
\end{align}
Moreover, the maps taking an element of $X$ to the element of the set~\eqref{eq_lem_aaa_ce} is a projection operator in $\Lp{\ffff[]}$, with values in $\Lp{\gggg}\cap \Phi$ and in $\Phi$, respectively.  
\end{lem}
\begin{proof}
The multifunction of \eqref{eq_lem_aaa_ce} is a metric projection onto the closed convex subset $\Phi\cup \Lp{\gggg}$, in the Hilbert space $\Lp{\ffff[]}$.  From the results of \citep[Chapter 29]{ConvexMonoCMS}, we conclude that it is well-defined and single-valued.  
\end{proof}
\begin{ex}[Variance Thresholding]\label{ex_proj_var}
The feature set $\Phi_{\Sigma}^2$, of Example~\ref{ex_Prespecified_Variance_Set}, is a closed ball about the origin in $\Lp{\ffff[]}$.  Hence, the metric projection $\Pi_{\Phi_{\Sigma}^2}$ onto the Ball of radius $\Sigma$ about $0$ is     \begin{equation}
\Pi_{\Phi_{\Sigma}^2}(Z) 
=
\min\left\{
1,\frac{\Sigma}{\sqrt{\ee{\|Z\|^2}}}
\right\}\cdot Z
\label{eq_met_proj_unit_ball}
.
\end{equation}
See \cite[Section 6.9]{devolder2014first} for a discussion about projections onto unit balls in Hilbert spaces.  
\end{ex}
\begin{ex}[Sparse Variance Thresholding]\label{ex_sparse_proj_ex}
Let $\Sigma>0$ and denote the $\ell^1$-ball in $\rrd$, about $0$ of radius $\Sigma$ by $\operatorname{Ball}_1(0;\Sigma)$.  

Since $\Lp{\ffff[]}[]$ is a decomposable space, in the sense of \citep[14.59 Definition]{rockafellar2009variational}, and $\|\cdot\|^2$ is a Carath\'{e}odory (also called a normal integrand in \cite{rockafellar2009variational}), therefore \citep[14.60 Theorem]{rockafellar2009variational} implies that
\begin{equation}
Z \in \arginf{Z \in \Phi^1_{\Sigma}}[1] \ee{\|X-Z\|^2}  
\Leftrightarrow
Z(\omega) \in \arginf{x \in \rrd}[1] \|X(\omega)-Z\|^2 ;\; (\forall \omega \in \Omega^{0})
\label{eq_interchange_rockafellar}
,
\end{equation}
where $\Omega^{0}\subseteq \Omega$ is a subset of full $\pp$-measure.  
That is, the projection $\Pi_{\Phi^1_{\Sigma}}$ may be computed $\pp$-a.s.~entirely in the image in $\rrd$.  Denoting the projection of $\rrd$ onto $\operatorname{Ball}_1(0;\Sigma)$ by $\Pi_1^{\Sigma}$,~\eqref{eq_interchange_rockafellar} implies that
\begin{equation}
\Pi_{\Phi^1_{\Sigma}}(X(\omega))=\Pi_1^{\Sigma}(X(\omega))
.
\label{ex_sparse_proj_ex_1}
\end{equation}
In \cite{duchi2008efficient}, it is shown that there exists some $\Sigma^{\star}\in (0,\infty)$ such that $\Pi_1^{\Sigma}$ is given explicitly by
\begin{equation*}
\Pi_1^{\Sigma}(x)\triangleq \sum_{n=1}^{d} \operatorname{sgn}(x_n)\left(x_n-\Sigma^{\star}\right)_+ e_n
,
\end{equation*}
where $\{e_n\}_{n=1}^d$ is the standard orthonormal basis of $\rrd$.  Therefore, $\Pi_{\Phi^1_{\Sigma}}$ may be defined $\pp$-a.s. for any $X=(X_1,\dots,X_d)\in \Lp{\ffff[]}[d]$ by
\begin{equation}
\Pi_{\Phi^1_{\Sigma}}(X)(\omega) = \sum_{n=1}^{d} \operatorname{sgn}\left(X_n(\omega)\right)\left(X_n(\omega)-\Sigma^{\star}\right)_+ e_n
.
\label{eq_form_pi_1r}
\end{equation}
The use of the word \textit{sparse} will be justified in the next section.  Intuitively, unlike~\eqref{eq_met_proj_unit_ball} the term $\left(X_n(\omega)-\Sigma^{\star}\right)_+$ can set some entries of $X$ to $0$.  
\end{ex}

\begin{defn}[{$\Phi$}-Relative Conditional Expectation]\label{defn_adm_cond_exp}
Let $\Phi$ be an feature set.  The $\Phi$-relative conditional expectation, is defined to be the map
\begin{equation}
\begin{aligned}
\ee{\cdot}[\gggg;\Phi]: \Lp{\ffff[]}&\rightarrow \Phi\cap \Lp{\gggg}\\
\ee{X}[\gggg;\Phi]&\mapsto \arginf{Z \in \Lp{\ffff[]}} 
\ee{\|X-Z\|^2}
+
\iota_{\Lp{\gggg}}(Z) + \iota_{\Phi}(Z) 
.
\end{aligned}
\label{defn_adm_cond_exp_eq_defining_equation}
\end{equation}
\end{defn}
\begin{ex}\label{rem_redux}
If $\Phi=\Lp{\ffff[]}$ and $\hhhh=\rrd$ then $\ee{\cdot}[\gggg;\Phi]=\ee{\cdot}[\gggg]$.  In the general case where $\hhhh\neq\rrd$, we will denote $\ee{\cdot}[\gggg;\Phi]$ by $\ee{\cdot}[\gggg]$.
\end{ex}
Other examples will be considered in latter sections.  

The problem of computing the intersection between closed-convex sets in a Hilbert space is solved in \cite{boyle1986method}.  The algorithm is known as Dykstra's method, or as Dykstra splitting within the scientific computing and signal processing communities respectively.  The next result is a direct application of Dykstra splitting.  
\begin{prop}[Computation of The {$\Phi$}-Relative Conditional Expectation]\label{prop_computable}
For any $X \in \Lp{\ffff}$, $\ee{X}[\gggg;\Phi]$ is the strong limit of the sequence $\left\{X^{(n)}\right\}_{n \in \nn}$, defined recursively by
\begin{equation}
\begin{aligned}
\bar{Y}^{(n)}\triangleq &\Pi_{\Phi}(\bar{X}^{(n)} + P^{(n)})\\
P^{(n+1)}\triangleq & \bar{X}^{(n)} + P^{(n)} - \bar{Y}^{(n)}\\
\bar{X}^{(n+1)} \triangleq & \ee{\bar{Y}^{(n)} + Q^{(n)}}[\gggg]\\
Q^{(n+1)}\triangleq &  \bar{Y}^{(n)} + Q^{(n)} - \bar{X}^{(n+1)}
,
\end{aligned}
\label{eq_prop_computable}
\end{equation}
where $\bar{X}^{(0)},P^{(0)}, Q^{(0)},$ and $\bar{Y}^{(0)}$ are initialized at
$$
\bar{Y}^{(0)}\triangleq 0;
\; P^{(0)}\triangleq 0 ;
\; \bar{X}^{(0)}\triangleq X;
\; Q^{(0)} \triangleq 0,
$$
where $\Pi_{\Phi}$ is the metric projection on $\Lp{\ffff[]}$ onto the set $\Phi$, and is defined by sending $X$ to the element in~\eqref{eq_lem_aaa_ce}.  
\end{prop}
\begin{proof}
Since conditional expectation, on $L^2$ is given by the projection onto the closed convex subset $\Lp{\gggg}$, and since $\Phi\cap \Lp{\gggg}\neq\emptyset$, then \citep[Theorem 30.7]{ConvexMonoCMS} implies the result.  
\end{proof}

We may now describe an algorithm for computing~\eqref{PG} and establish its convergence.  
\subsection{Forward-Backwards Splitting}\label{s_FB_Split}
If $\riskutility$ is G\^{a}teaux differentiable, then we may define the sequence $\{Z^n\}_{n \in \nn}$ in $\Lp{\gggg}$ by
\begin{enumerate}[(i)]
\item $Z^0\triangleq \ee{X}[\gggg;\Phi][\qq]$,
\item For all $n>0$ let 
$$Z^{n+1}\triangleq 
Z^n + \alpha_n
\left(
\ee{
	Z^n  +%
	(2\gamma_n \tilde{\lambda})\cdot %
	(\nabla \riskutility
	\left(
	X
	-
	Z^n
	\right)  +B_n)
}[\gggg;\Phi][\qq]
+A_n- Z^n
\right)
.
$$
\end{enumerate}
The next result, shows that under appropriate choices of the meta-parameters $\{\gamma_n,\alpha_n,B_n,A_n\}_{n \in \nn}$, $Z^{n}$ converges to the solution to Problem~\ref{PG}.  The type of convergence depends on the convexity of $\riskutility$. 
\begin{ass}[Compatibility]\label{ass_convergence}
Let $\beta>0$, $\{\gamma_n,\alpha_n,B_n,A_n\}_{n \in \nn}$ be a sequence taking values in $(0,\infty)\times (0,1]\times \Lp{\gggg}^2$, such that
\begin{enumerate}[(i)]
	\item $0<\inf_{n \in \nn} \gamma_n <\sup_{n \in \nn} \gamma_n <2\beta$,%
	\item $\inf_{ n \in \nn} \alpha_n >0$,
	\item $\sup_{n \in \nn}\alpha_n \leq 1$,
	\item $
	\max\left\{
	\sum_{n\in \nn} \ee{\|A_n\|^2}
	,
	\sum_{n\in \nn} \ee{\|B_n\|^2}
	\right\}
	<
	\infty
	.
	$
\end{enumerate}
Moreover, assume that $\riskutility$ is G\^{a}teaux differentiable on $\Lp{\ffff[]}$.  
\end{ass}
This result follows nearly directly from the Proximal Forwards-Backwards Splitting algorithm of \cite{combettes2005signal}.  
\begin{thrm}[Proximal Forwards-Backwards Splitting (PFBS)]\label{thrm_compute}
Under Assumption~\ref{ass_convergence},    $\{Z^n\}_{n \in \nn}$ converges weakly to $\eeRAp{X}$.  If moreover, $\riskutility$ is strictly-convex with Lipschitz constant $\frac1{\beta}>0$, then $\{Z^n\}_{n \in \nn}$ converges strongly to $\eeRAp{X}$.  
\end{thrm}
\begin{cor}\label{cor_simplified_algo_gen}
Fix $\gamma>2\beta$.  
If $\riskutility$ is a strictly convex and G\^{a}teaux differentiable with Lipschitz constant $\frac1{\beta}$.  Then, the $\Lp{\gggg}$-valued sequence $\{Z^n\}_{n\in\nn}$, defined by 
$$
\begin{aligned}
Z^0\triangleq & \ee{X}[\gggg;\Phi]\\
Z^{n+1}\triangleq & 
\ee{
	Z^n}[\ffff[0];\Phi][\qq]
- \frac{\tilde{\lambda}}{\gamma^n}
\ee{
	\nabla \riskutility
	\left(
	X
	-
	Z^n
	\right) 
}[\gggg;\Phi][\qq]
;\qquad (n\geq 1)
,
\end{aligned}
$$
converges strongly to $\eeRAp{X}$, where $\lambda'\triangleq \tilde{\lambda}2 \beta$.    
\end{cor}
\begin{proof}
Take $A_n=B_n=0$, $\gamma_n\triangleq \frac1{\gamma^n}$, and $\alpha_n=1$, for every $n \in \nn$ in Theorem~\ref{thrm_compute}.  
\end{proof}
If all estimators are $\Phi$-featured, as is the case in the classical setting, the use of Dykstra's method can be avoided.  In this case, the algorithm of Corollary~\ref{cor_simplified_algo_gen} reduces to.  
\begin{cor}\label{cor_simplified_algo}
	Set $\Phi=\Lp{\ffff[]}$ and fix $\gamma>2\beta$.    If $\riskutility$ is a strictly convex and G\^{a}teaux differentiable with Lipschitz constant $\frac1{\beta}$.  Then, the sequence $\{Z^n\}_{n\in\nn}$ in $\Lp{\gggg}$, defined by 
	$$
	\begin{aligned}
	Z^0\triangleq & \ee{X}[\gggg]\\
	Z^{n+1}\triangleq & 
	\ee{
		Z^n}[\ffff[0]][\qq]  
	- \frac{\tilde{\lambda}}{\gamma^n}
	\ee{
		\nabla \riskutility
		\left(
		X
		-
		Z^n
		\right) 
	}[\gggg][\qq]
	;\qquad (n\geq 1)
	,
	\end{aligned}
	$$
	converges strongly to $\eeRAp{X}$, where $\lambda'\triangleq \tilde{\lambda}2 \beta$.    
\end{cor}
\begin{proof}
	The result follows from Corollary~\ref{cor_simplified_algo} and Example~\ref{rem_redux}.  
\end{proof}

\begin{proof}[Proof of Theorem~\ref{thrm_compute}]
	First note, that the proximal operator of the indicator function $\iota_{\Lp{\gggg}}$ is the conditional expectation $\ee{\cdot}[\gggg]$.  
	
	Theorems~\ref{thrm_existsnce_uniqueness} and~\ref{thrm_Lagrangian_Formulation}, guarantees that the set of optimizers of Problem~\eqref{PG} is non-empty.  Therefore, Assumption~\ref{ass_convergence} guarantees that \citep[Theorem 3.4]{combettes2005signal}(i) holds; from which the weak converges follows. 
	
	Under the additional assumption that $\riskutility$ is strictly convex, \citep[Proposition 3.6]{combettes2005signal}(vii), implies that \citep[Condition 3.2]{combettes2005signal} holds.  Hence, \citep[Theorem 3.4]{combettes2005signal}(iv) implies the strong convergence of $\{Z^{(n)}\}$ to $\rrrr(X)$.  
\end{proof}

\begin{ex}[Quadratic Risk-Measure]\label{ex_QRM}
	Building on Examples~\ref{lem_Qudratic_risk} and~\ref{ex_dist_target} we find that
	\begin{equation}
	\nabla {\rho_2}%
	(-\|X-Z\|^2) %
	= \left(X-Z\right)\cdot\left(\|X-Z\|^2 -\left(\ee{\|X-Z\|^2} -%
	\frac1{2}\right)\right)
	.
	\label{ex_QRM_Gateau_derivative_riskutility}
	\end{equation}
	
	If $\hhhh=\rrd$, the target $X$ is taken to be the value of the derivative on $X\ledot$, with maturity $T>0$, and Borel-measurable payoff function $f$, then using Equation~\eqref{ex_QRM_Gateau_derivative_riskutility} in Corollary~\ref{cor_simplified_algo} leads to the following procedure for computing the solution to the risk-averse valuation problem~\ref{RA}
	\small
	$$
	\begin{aligned}
	Z^0\triangleq & \ee{f(X_T)}[\ffff[0]][\qq]\\
	Z^{n+1}\triangleq & 
	\ee{
		Z^n  - \frac{\lambda}{\gamma^n} 
		\left(f(X_T)-Z^n\right)^{\star}\left(\|f(X_T)-Z^n \|^2 -\left(\ee{\|f(X_T)-Z^n \|^2} -\frac1{2}\right)\right)
	}[\ffff[0]][\qq]
	,
	\end{aligned}
	$$
	\normalsize
	for the risk-aversion level $M_{\lambda'}$, where $\lambda'\triangleq\tilde{\lambda}2 \beta$.    
\end{ex}

Before discussing the algorithmic computation of the risk-averse conditional expectation, we highlight some connections between $\rrrr$-conditioning, robust finance, sensitivity analysis, risk-measures, and high-dimensional probability theory.  
\section{Connections With Other Literature}\label{s_rel_oLit}
\subsection{Connections to High-Dimensional Statistics and Machine Learning}\label{sss_sce}
Let $1<d<D$ be a positive integers.  Define the function $\|\cdot\|_0:\rrd[D]\rightarrow \{0,\dots,D\}$ as counting the number of zero entries of a vector.  Following the high-dimensional statistical literature, such as \cite{tibshirani1996regression,zou2006sparse,bioucas2012hyperspectral}, we will say that a vector $x \in \rrd[D]$ is \textit{sparse} if
\begin{equation}
\|x\|_0<D
\label{eq_sparse_rrd}
.
\end{equation}
Geometrically, condition~\eqref{eq_sparse_rrd} expresses the fact that $x$ lies in a low-dimensional linear subspace $\rrd$ of $\rrd[D]$, where $d= \|x\|_0$.  We will extend the property expressed by~\eqref{eq_sparse_rrd} to square-integrable random vectors as follows.

We begin by viewing $\Lp{\ffff[]}[d]$ as a subspace of $\Lp{\ffff[]}[D]$ via the canonical embedding
$$
\begin{aligned}
\Lp{\ffff[]}[d]&\hookrightarrow \Lp{\ffff[]}[D]\\
(X_1,\dots,X_d)&\mapsto (X_1,\dots,X_d,0,\dots,0)
.
\end{aligned}
$$
Next, we construct the subsets $\SparseSpace$ of $\Lp{\ffff[]}[D]$, whose members have at-least a probability of $\epsilon$ of being having at most $d$ non-zero entries.  Formally these are defined as follows.  
\begin{defn}[The {$\epsilon$}-Probably Sparse Space, {$\SparseSpace$}]\label{defn_sparse_spaces}
	Fix $\epsilon \in [0,1]$ and $D\in\nn$.  We define the set of $\epsilon$-probably sparse random-vectors, to be the subset $\SparseSpace$ of $\Lp{\ffff[]}[D]$ defined by
	$$
	\SparseSpace\triangleq
	\left\{
	X \in \Lp{\ffff[]}[D] :\,
	\pp\left(
	\|X\|_0\leq d
	\right)\geq \epsilon
	\right\}.
	$$
\end{defn}
\begin{prop}[Properties of {$\SparseSpace$}]\label{prop_Basic_probs_Space_Spaces}\hfill
	\begin{enumerate}[(i)]
		\item $\SparseSpace[0]=\SparseSpace[\epsilon][\ffff[]][D][D]=\Lp{\ffff[]}$,
		\item $\SparseSpace[1]=\left\{
		X \in \Lp{\ffff[]}:\|X\|_0\leq d, \, \pp\mbox{-}a.s.
		\right\},
		$
		\item If $0\leq \epsilon_1\leq \epsilon_2\leq 1$ then
		$$
		\begin{aligned}
		\SparseSpace[\epsilon_2]\subseteq \SparseSpace[\epsilon_1]
		,
		\end{aligned}
		$$
		\item If $0\leq d_1\leq d_2\leq D$ then
		$$
		\begin{aligned}
		\SparseSpace[\epsilon][\ffff[]][D][d_1]
		\subseteq 
		\SparseSpace[\epsilon][\ffff[]][D][d_2]
		,
		\end{aligned}
		$$
		\item The map $X \mapsto \epsilon^{X}_d\triangleq \max\left\{
		\epsilon \in [0,1]: X \in \SparseSpace
		\right\}$, is well-defined,
		\item The space $\SparseSpace$ is star-shaped in $\Lp{\ffff[]}$, with center $0$,
		\item The space $\SparseSpace$ is not convex, if $\epsilon>0$ and $0<d<D$
		.
	\end{enumerate}
\end{prop}
\begin{proof}
	Properties (i)-(ii) follow from the definition of $\SparseSpace$.  For Property (iii), note that if $X \in \SparseSpace[\epsilon_2]$ then
	$$
	\pp\left(\|X\|_0\leq d\right)\geq \epsilon_2\geq \epsilon_1,
	$$
	hence $X \in \SparseSpace[\epsilon_1]$.  
	For property (iv), simply note that for every $\omega \in \Omega$, and every $0\leq d_1\leq d_2\leq D$ it follows that
	$$
	\|X(\omega)\|_0 \leq d_1\leq d_2.  
	$$
	To see Property (v), take $\epsilon^{X}_d\triangleq \pp\left(\|X\|_0\leq d\right)$, and notice that this achieves the maximum of the set
	$$
	\max\left\{
	\epsilon \in [0,1]: X \in \SparseSpace
	\right\}
	.
	$$

	For Property (vi), first observe that $0 \in \SparseSpace[1]$ and that for every $k \in \rr-\{0\}$ and every $\omega \in \Omega$,
	$$
	\|k\cdot X(\omega)_0\|=\|X(\omega)\|_0.  
	$$
	Hence $k\cdot X \in \SparseSpace$ for every $k \in \rr$, if $X \in \SparseSpace$.

	To see Property (vii), consider the scaled sum $k\cdot X^A + r\cdot X^B$, where the random vectors $X^A$ and $X^B$ are defined for each $\omega \in \Omega$, $k,r \in \rr-\{0\}$, and each\\
	$i \in \{1,\dots,D\}$ by
	$$
	(X^A)^i(\omega) \triangleq
	(X^B)^{D-i+1}(\omega) \triangleq 
	\begin{cases}
	1 &: \mbox{ if }i \leq d\\
	0 &: \mbox{ if } else
	.
	\end{cases}
	$$
	By construction $X^A,X^B \in \SparseSpace$ but $X^A+X^B$ is $\pp$-a.s. never sparse, even in the case where $k=(1-r)$ and $r \in (0,1)$.

\end{proof}
Tough intuitive, Proposition~\ref{prop_Basic_probs_Space_Spaces} shows that the spaces $\SparseSpace$ are generally not overly well-behaved.  Nevertheless, the next result shows that for a certain formulation of Problem~\eqref{PG}, the corresponding $\rrrr$-conditioning operator produces elements of $\SparseSpace$.
\begin{prop}[Sparse Conditional Expectation]\label{prop_sparse_cond_exp}
	Let $0<d<D$, $\Sigma>0$ and take $\delta \in [0,1)$%
	.  If 
	$X \in \SparseSpace[\delta][\gggg] - \SparseSpace[1][\ffff[]][D][D]$.  %
	Then there exists some $\Sigma>0$ and some $\epsilon \in (\delta ,1]$ for which
	$$
	\ee{X}[\gggg;\Phi^1_{\Sigma}] \in \SparseSpacesigma
	\mbox{ and }
	X \not\in \SparseSpacesigma
	.  
	$$
\end{prop}
The proof of Proposition~\ref{prop_sparse_cond_exp} relies on the following technical lemma.      
\begin{lem}[Probabilistic Sparsity]\label{lem_sparse_dr}
	Let $X \in \Lp{\ffff[]}[D]$, $0<d<D$, and $X \in \SparseSpace[\delta]$ and suppose that $\max_{i=1,\dots,D}\pp\left(X_i\neq 0\right)<1$. Then there exists $\Sigma>0$, such that 
	$$
	\Pi_{\Phi^1_{\Sigma}}\left(X\right) \in \SparseSpace
	\mbox{ and }
	X \not\in \SparseSpace
	,
	$$ and $\epsilon >\delta$.  
\end{lem}
\begin{proof}
	Let $\delta\triangleq \epsilon^X_d$.  If $X\neq 0$, $\pp$-a.s.~then there exists some $\Sigma>0$ such that for every $n\in \{1,\dots,D\}$, the non-triviality condition
	\begin{equation}
	\pp\left((\forall n \in \{1,\dots,D\})\, 0<X_n \leq \Sigma^{\star}
	\right)>0
	\label{eq_non_trivl}
	,
	\end{equation}
	holds, where $\Sigma^{\star}$ is defined as in~\eqref{eq_form_pi_1r}.      
	
	In order to be tidy, set $B\triangleq \overline{\operatorname{Ball}_1(r;0)}-\{0\}$.  Define the events $\xxxxx$ and $\yyyyy$ by
	\begin{align}
	\nonumber
	\xxxxx\triangleq & \left\{
	\omega \in \Omega : \|X(\omega)\|_0 \leq d
	\right\} \\
	\nonumber
	=&
	\left\{
	\omega \in \Omega : \left(
	\exists i \in \{D,\dots,D-d\}
	\right)\left(\exists (n_1,\dots,n_i) \in {\binom{S(D)}{i}}\right)
	\right.\\&\left.
	0=X_{n_1}(\omega)=\dots=X_{n_i}(\omega)
	\right\}
	,
	\\
	\nonumber
	\yyyyy\triangleq & \left\{
	\omega \in \Omega : X(\omega) \in B
	\right\},
	\end{align}
	where $\binom {S(D)}n$ is the \textit{set} of all $n$-combinations of the set $S(D)\triangleq \{1,\dots,D\}$.  By the definition of $B$ and the non-triviality condition~\eqref{eq_non_trivl}, it follows that
	\begin{equation}
	\pp\left(\yyyyy
	\right)>0 \mbox{ and } %
	\yyyyy \cap \xxxxx%
	=%
	\emptyset
	.
	\label{eq_condits}
	\end{equation}
	We will show that $\Pi_{\Phi^1_{\Sigma}}(X)(\omega)$ for every $\omega \in \xxxxx\cup\yyyyy$.  
	
	We first show that $\Pi_{\Phi^1_{\Sigma}}(X)$ is at least as sparse as $X$, on $\xxxxx$.  By~\eqref{eq_form_pi_1r}, it follows that for every $\omega \in \Omega$ if $X_n(\omega)=0$, for some $n \in \{1,\dots,D\}$, then
	$$
	\Pi_{\Phi^1_{\Sigma}}(X)_n(\omega) = \operatorname{sgn}\left(X_n(\omega)\right)\left(X_n(\omega)-\Sigma^{\star}\right)_+ e_n =0\,e_n.
	$$
	Hence, $\|X\|_0\leq d$ implies that $\|\Pi_{\Phi^1_{\Sigma}}(X)\|_0\leq d$, $\omega$-wise.  Therefore, 
	\begin{equation}
	\xxxxx\subseteq 
	\left\{
	\omega \in \Omega :
	\|\Pi_{\Phi^1_{\Sigma}}(X)(\omega)\|_0\leq d
	\right\}.
	\label{eq_first_piece_sparsity}
	\end{equation}
	
	Moreover,~\eqref{eq_form_pi_1r} implies that for every $\omega \in \yyyyy$
	$$
	\Pi_{\Phi^1_{\Sigma}}(X)(\omega) = \sum_{n=1}^{d} \operatorname{sgn}\left(X_n(\omega)\right)\left(X_n(\omega)-\Sigma^{\star}\right)_+ e_n
	= \sum_{n=1}^{d} 0\, e_n
	.
	$$
	Therefore,
	\begin{equation}
	\yyyyy\subseteq 
	\left\{
	\omega \in \Omega :
	\|\Pi_{\Phi^1_{\Sigma}}(X)(\omega)\|_0\leq d
	\right\}.
	\label{eq_second_piece_sparsity}
	\end{equation}
	
	Combining~\eqref{eq_condits},~\eqref{eq_first_piece_sparsity} and~\eqref{eq_second_piece_sparsity}, we conclude that
	\begin{equation}
	\pp\left(\|X\|_0\leq d\right)
	< \pp\left(\xxxxx\right)+\pp\left(\yyyyy\right)=\pp\left(\xxxxx \cup \yyyyy\right)\leq 
	\pp\left(
	\|\Pi_{\Phi^1_{\Sigma}}(X)\|_0\leq d
	\right)
	.
	\end{equation}
	Applying the definitions of $\SparseSpace[\delta]$ and $\SparseSpace$, we obtain the result.  
\end{proof}
\begin{proof}[Proof of Proposition~\ref{prop_sparse_cond_exp}]
	By definition of $\ee{X}[\gggg;\Phi^1_{\Sigma}]$ (see Lemma~\ref{lem_aaa_ce}), since $X\in \Lp{\gggg}[D]$, it follows that $\ee{X}[\gggg;\Phi^1_{\Sigma}]=\Pi_{\Phi^1_{\Sigma}}(X)$.  Applying Lemma~\ref{lem_sparse_dr} yields the conclusion.  
\end{proof}

\subsection{Connection to Robust Finance}\label{ss_crf}
For the remainder of this paper, let $X_t$ be an $\rr$-valued (price) process, which is $\pp$-a.s$.$ positive.

Robust finance is concerned with estimation and prediction of financial quantities, under the relaxation of the assumption that the correct model is known.  Typically, this is expressed by regarding the probability measure as being unknown and belonging to a broader set of possible, alternative probability measures.  

For any $Q\ll\qq$, such that $\frac{dQ}{d\qq} \in\Lp{\ffff[]}[]$, the density process $Z_t\triangleq \ee{\frac{dQ}{d\qq}}[\ffff[t]]$ defines a stochastic process $X_t^{Q}$ by
$$
X_t^{Q}\triangleq Z_tX_t.
$$
The process $X_t^{Q}$, is typically interpreted as describing the dynamics of $X_t$ under $\qq$.  Let us consider a relaxation of this perspective.  For us, alternative models, are processes which at time $T$ are elements of the set
$$
\mmmm_{alt}^T(X)\triangleq \left\{
Z_T \in \Lp{\ffff[T]}[] :
(\exists Y_T \in \Lp{\ffff[T]}[]) \; Z_T=Y_TX_T
\right\}
.
$$

We now show that the optimal value of the problem~\eqref{defn_RA_Expectation_eq_defining_alt}, defined by
\begin{equation}
\rho_{\lambda,t}\triangleq (1-\lambda)\ee{
	\left\|
	X_t-\rrrr(X_t)
	\right\|^2
}
+\lambda
\riskutility(X_t-\rrrr(X_t))
,
\label{eq_optim_val}
\end{equation}
admits a robust representation analogous to that of \cite{delbaen2002coherent,KainaCRMs,cheridito2008dual,cohenqs}.  Notice that, the optimal value $\rho_{\lambda,t}$ is the infimal convolution of $f$ and $g$, where $f$ and $g$ are defined as in Remark~\ref{rem_IC_rep}.  

\begin{prop}[Robust Representation]\label{prop_rob}
	Suppose that $X_t\neq 0$, $\pp$-a.s. and that\footnote{
		Recall that,
		$
		\operatorname{sri}(C)\triangleq \left\{
		x \in C : cone(C-x) = \overline{
			\operatorname{span}\left(
			C-x
			\right)
		}
		\right\}
		$,
		where $cone(C)$ is the smallest convex cone containing the set $C$.      
	}
	\begin{equation}
	0 \in \operatorname{sri}\left(
	\operatorname{dom}\left(
	\left[\lambda\riskutility(\cdot) + \ee{\|\cdot \|^2}\right]
	\right)
	+
	\operatorname{dom}\left(\iota_{\Lp{\gggg}[]\cap \Phi}(\cdot)\right)
	\right)
	\label{regul}
	.
	\end{equation} 
	The functional $\rho_{\lambda,t}$ admits the following robust representation
	\begin{equation}
	\begin{aligned}
	\rho_{\lambda,t}%
	=&
	\sup_{\tilde{X}_t\in \mmmm_{alt}^t(X)}\ee{\tilde{X}_t} - \alpha\left(\frac{\tilde{X}_t}{X_t}\right)\\
	\alpha(Z)\triangleq &
	\left[
	\left(\lambda\riskutility(\cdot) + \ee{\|\cdot \|^2}\right)^{\star}\right]\left(Z\right)
	- \sigma_{\Lp{\gggg}[]\cap \Phi}(Z)
	,
	\end{aligned}
	\label{eq_rob_rep_le}
	\end{equation}
	here $\sigma_C$ is typically called the support function of the set $C$, and is defined by
	$$
	\sigma_C(Z)\triangleq \sup\{ \ee{ZY}: Y\in C\}
	.
	$$
\end{prop}
\begin{ex}\label{ex_sri}
	If $\rho$ is finite-valued on all of $\Lp{\ffff[]}[]$ and $\Phi$ is a closed linear subspace of $\Lp{\gggg}[]$, then condition~\eqref{regul} must hold.  
\end{ex}
\begin{proof}[Proof of Proposition~\ref{prop_rob}]
	Denote the set of all convex, lsc, proper maps from $\Lp{\ffff[]}[]$ to $\rrext$ by $\Gamma_0(\Lp{\ffff[]}[]$.  %
	The map $\rho_{\lambda,t}$ is given by the infimal convolution
	\begin{equation}
	\rho_{\lambda,t}(\cdot)=\left[\lambda\riskutility(\cdot) + \ee{\|\cdot \|^2}\right] \square \iota_{\Lp{\gggg}[]\cap \Phi}
	\label{eq_inf_convl}
	.
	\end{equation}
	
	Since both $\left(\lambda\riskutility(\cdot) + \ee{\|\cdot \|^2}\right)$ and $\iota_{\Lp{\gggg}[]\cap \Phi}$ are in 
	$\Gamma_0(\Lp{\ffff[]}[]$, the \citep[Fenchel-Moreau Theorem; 13.37]{ConvexMonoCMS} implies that~\eqref{eq_inf_convl_simpl} may be rewritten as
	\begin{equation}
	\begin{aligned}
	\rho_{\lambda,t}(\cdot)
	=&\left(\lambda\riskutility(\cdot) + \ee{\|\cdot \|^2}\right)^{\star\star} \square \iota_{\Lp{\gggg}[]\cap \Phi}^{\star\star}\\
	=&\left(\lambda\riskutility(\cdot) + \ee{\|\cdot \|^2}\right)^{\star\star} + \sigma_{\Lp{\gggg}[]\cap \Phi}^{\star}
	.
	\end{aligned}
	\label{eq_inf_convl_simpl}
	\end{equation}
	
	Assumption~\eqref{regul} guarantees that the infimal convolution is exact; therefore \cite[the Attouch-Br\'{e}zis Theorem; 15.3]{ConvexMonoCMS} may be applied to~\eqref{eq_inf_convl_simpl}, thus
	\begin{equation}
	\begin{aligned}
	\rho_{\lambda,t}(X)
	=&\left[
	\left(\lambda\riskutility(\cdot) + \ee{\|\cdot \|^2}\right)^{\star} + \sigma_{\Lp{\gggg}[]\cap \Phi}
	\right]^{\star}(X)\\
	=&\sup_{Z \in \Lp{\ffff[]}[]}\ee{ZX} - 
	\left[
	\left(\lambda\riskutility(\cdot) + \ee{\|\cdot \|^2}\right)^{\star}\right](Z)\\
	&- \sigma_{\Lp{\gggg}[]\cap \Phi}(Z)
	.
	\end{aligned}
	\label{eq_inf_convl_simpl_2}
	\end{equation}
	
	Setting $X\triangleq X_T$, note that the elements of $\mmmm_{alt}^t(X)$ are all exactly of the form $ZX_T$,hence~\eqref{eq_inf_convl_simpl_2} can be rewritten as
	\begin{equation}
	\begin{aligned}
	\rho_{\lambda,t}(X_t)
	=&
	\sup_{\tilde{X}_t\in \mmmm_{alt}^t(X)}\ee{\tilde{X}_t} - 
	\left[
	\left(\lambda\riskutility(\cdot) + \ee{\|\cdot \|^2}\right)^{\star}\right]\left(\frac{\tilde{X}_t}{X_t}\right)\\
	&- \sigma_{\Lp{\gggg}[]\cap \Phi}\left(\frac{\tilde{X}_t}{X_t}\right)
	,
	\end{aligned}
	\label{eq_inf_convl_simpl_3}
	\end{equation}
	which is well-defined $\pp$-a.s. since $X_t\neq 0$ $\pp$-a.s. This gives~\eqref{eq_rob_rep_le}.  
\end{proof}
Proposition~\ref{prop_rob} is used to establish a connection to sensitivity analysis.  
\subsection{Connections to Sensitivity Analysis}
A sensitivity, also called a Greek, is a partial derivative of the expected risk-neutral payoff of a financial derivative on the underlying $X_t$, with respect to a parameter.  Typical choices are sensitivities to the initial value of the process $X_t$, or the process' volatility.  These sensitivities, studied for example in \cite{di2009malliavin,el2004computations}, are typically computed using the Bismut-Elworthy-Li formula as in~\cite{takeuchi2010bismut,cass2006bismut}, based on the Malliavin calculus studied in~\cite{nualart1986generalized,nualart2006malliavin} or through the Functional It\^{o} calculus of~\cite{dupire2009functional,fournie2010functional}, as in~\cite{jazaerli2017functional}.  However, other approaches have also been considered, for example, see \cite{kusuoka2004approximation,teichmann2005calculating}.  

A portfolio whose price is modeled by $X_t$, can be used to construct a portfolio which is neutral to a specific parameter, which we will denote by $\beta$.  This $\beta$-neutral portfolio can be defined by
\begin{equation}
X_t - \nabla_{\beta}\ee{X_t}[\gggg_t].
\label{eq_sen_neutral}
\end{equation}
For example, when $\beta$ is taken to be the initial value of $X_t$,~\eqref{eq_sen_neutral} is precisely the definition of a Delta-neutral portfolio.  
Our next result shows that certain $\rrrr$-conditioning operators, can be interpreted via a robust analogue of~\eqref{eq_sen_neutral}.  
\begin{cor}[Robust-Neutral]\label{prop_rob_neutral}
	Assume the setting of Example~\ref{ex_Lagrangian_mRC} and set\\ $\rho=R$.  
	Therefore, there exists $\eta>0$ such that $\rrrr^0(X)$ solves problem~\eqref{eq_constrianed}; that is,
	\begin{equation}
	\rrrr^0(X) = \arginf{Z \in \Lp{\gggg}[]}[] \frac1{2}\ee{\|X-Z\|^2} + \eta \rho(g(Z))%
	.
	\label{eq_lagrangian}
	\end{equation}
	
	Then $\rrrr^0(X_t)$ admits the following representation
	\begin{equation}
	\rrrr^0(X_t)=X_t - \eta\nabla\left(
	\sup_{\tilde{X}_t\in \mmmm_{alt}^t(X)}\ee{\tilde{X}_t} - \alpha\left(\frac{\tilde{X}_t}{X_t}\right)
	\right)
	\label{eq_uncertainty_neutral}
	,
	\end{equation}
	where $\nabla$ is the G\^{a}teaux derivative in $\Lp{\ffff[]}$.  
\end{cor}
\begin{proof}
	Under Assumption that~\eqref{eq_lagrangian} holds, it follows from Example~\eqref{ex_accept_RM} that 
	$$
	\rrrr^0(X) %
	= \operatorname{Prox}_{\gamma\rho\circ g +\iota_{\Lp{\gggg}}}(X)
	$$
	and that $\rho_{\lambda,t}(X_t)$ is the Moreau-Yoshida envelope of $\riskutility\circ g + \iota_{\Lp{\gggg}[]\cap \Phi}(\cdot)$ (see Definition\citep[Definition 12.20,Definition 12.23]{ConvexMonoCMS}) evaluated at $X_t$.  Applying \citep[Proposition 12.30]{ConvexMonoCMS} and Proposition~\ref{prop_rob} to~\eqref{eq_inf_convl_simpl_3} gives~\eqref{eq_uncertainty_neutral}.  
\end{proof}

\subsection{Regular Approximate Extensions of Risk-Measures}\label{ss_risk_measure_extension}
In \cite{delbaen2002coherent}, it was shown that any proper, lsc, convex risk-measure on $\Lp{\ffff[]}[][p]$, for $1< p< \infty$, admits the following robust representation, 
\begin{equation}
\begin{aligned}
\rho(X)= & \sup_{Z \in \cccc_p}\ee{ZX} -\rho^{\star}(Z),\\
\cccc_p\triangleq &
\left\{
Z \in \Lp{\ffff[]}[][p]^{\star}_- : \ee{Z}=-1
\right\}
,
\end{aligned}
\label{eq_extension}
\end{equation}
where $\rho^{\star}$ is the Fenchel-Legendre conjugate of $\rho$ on $\Lp{\ffff[]}[][p]$, $\Lp{\ffff[]}[][p]^{\star}_-$ denotes the lower-orthant of the dual space of $\Lp{\ffff[]}[][p]^{\star}$, and the set $\cccc_p$ is interpreted in \cite{delbaen2002coherent} as the collection of "\textit{generalized [market] scenarios}".

In \cite{filipovic2007convex}, the problem of extending a convex risk-measure $\rho$ from $\Lp{\ffff[]}[][p]$ to $\Lp{\ffff[]}[][r]$, for $1 \leq r<p\leq \infty$ was explored.  By exploiting~\eqref{eq_extension}, the authors define the extended risk-measure by
\begin{equation}
\begin{aligned}
\rho^r(X)\triangleq & \sup_{Z \in \cccc_r}\ee{ZX} -%
\rho^{\star}%
(Z) 
.
\end{aligned}\label{eq_extension_Filipovic_Svindland}
\end{equation}

\begin{remark}\label{rem_ordering}
	Note, that in the case where $r=2$, the second convex conjugate is taken in a strictly smaller space, since $\frac{p}{p-1}\leq 2$ if $2\leq p$.  
\end{remark}

In \citep[Theorem 3.4]{filipovic2007convex}, it was shown that if $\rho$ is law-invariant, then 
$$
\rho^r|_{\Lp{\ffff[]}[][p]}=\rho.
$$
In general, however, the risk-measures $\rho^r$ and $\rho$ need not coincide on $\Lp{\ffff[]}[][p]$.  

We the case where $r=2$.  The $\rrrr$-conditioning operator may be used to formulate a family of approximate extensions $\{\rho_{\lambda}\}_{\lambda>0}$, of $\rho$ to $\Lp{\ffff[]}[]$.  This family takes finite values on all of $\Lp{\ffff[]}[]$, be Fr\'{e}chet differentiable, and approximate $\rho$ to arbitrary precision on its domain.  However, they are no-longer be risk-measures, and instead, they provide a well-behaved approximation of $\rho$.  

Define a naive extension of a function $f:\Lp{\ffff[]}[][p]\rightarrow (-,\infty,\infty]$ by
\begin{equation}
f^{+}(X)\triangleq \begin{cases}
f(X) &:  X \in \operatorname{dom}(f)\\
\infty &:  \mbox{else}
.
\end{cases}
\label{eq_naive_ext}
\end{equation}
Similarly to~\eqref{eq_optim_val}, using the notation and assumptions of Example~\eqref{ex_Lagrangian_mRC}, define the family of functions $\left\{\rho_{+:\eta}\right\}_{\eta>0}$ by
\begin{equation}
\begin{aligned}
\rho_{+:\eta}(X)= & \rho^{+}\left(\rrrr_{\eta}(X)\right) + %
\frac1{2\eta}
\ee{\left(X-\rrrr_{\eta}(X)\right)^2}\\
\rrrr_{\eta}\triangleq & \left(
\ffff,\rho^{+},\utility%
,\aaaa_{\rho^{+}}^{1_{\rr},m},M_0
\right)
,
\end{aligned}\label{eq_extension_approx}
\end{equation}
and $\utility$ is any compatible utility operator.  
Then $\left\{\rho_{+:\eta}\right\}_{\eta>0}$ possesses the following properties.  

\begin{prop}[Regular Approximate Extension of {$\rho$}]\label{prop_approx_ext}
	Under assumption that~\eqref{eq_super_powers} and~\eqref{eq_super_powers_alt_Lagrangian} are equal, the map $\rho_{+:\eta}$ is a real-valued, lsc, convex proper function on $\Lp{\ffff[]}[]$ satisfying:
	\begin{enumerate}[(i)]
		\item\textbf{Full Domain:} $\operatorname{dom}(\rho_{+:\eta})=\Lp{\ffff[]}[]$,
		\item\textbf{Smooth:} $\rho_{+:\eta}$ is Fr\'{e}chet differentiable.  Moreover, its gradient is equal to
		$$
		\nabla \rho_{+:\eta}(X) = \frac1{\eta}\cdot \left(X-\rrrr_{\eta}(X)\right),
		$$
		\item\textbf{Approximation:} $\lim\limits_{\eta\downarrow 0} \rho_{+:\eta}(X)=\rho^{+}(X)$.  \\
		In particular, 
		if $X \in \operatorname{dom}(\rho)$, then $\lim\limits_{\eta\downarrow 0} \rho_{+:\eta}(X)=\rho(X)$,
		\item\textbf{Relation to The Extension{~\eqref{eq_extension_Filipovic_Svindland}}:} For every $X \in \Lp{\ffff[]}[]$ and every\\ $\eta \in (0,\infty)$, the risk-measure $\rho_2$ is related to $\rho_{+:\eta}$ through their Fr\'{e}chet-Legendre conjugates via the following formula
		$$
		\left(\rho_2\right)^{\star}
		=
		\rho_{+:\eta}^{\star}(X)|_{\Lp{\ffff[]}[]}
		-
		\frac{\eta%
		}{2%
		}
		\ee{
			X^2
		}
		.
		$$
	\end{enumerate}
\end{prop}
The proof of Proposition~\ref{prop_approx_ext}, relies on the following technical Lemma.
\begin{lem}\label{lem_conj}
	Suppose that $f:\Lp{\ffff[]}[][p]$ and $p\geq 2$, then $\left(f^{+}\right)^{\star}=f^{\star}$.  
\end{lem}
\begin{proof}
	For all $X \in \Lp{\ffff[]}[]$,
	\begin{align}
	\left(f^{+}\right)^{\star}(X) = & \sup_{Z \in \Lp{\ffff[]}[][p]}\ee{ZX} - f^{+}(Z).
	\label{eq_optim_simpl}
	\end{align}
	However, if $Z \in \Lp{\ffff[]}[][2]-\Lp{\ffff[]}[][p]$ then $-f^{+}=-\infty$, hence the optimization of~\eqref{eq_optim_simpl} may be restricted to $\Lp{\ffff[]}[][p]$.  Thus~\eqref{eq_optim_simpl} simplifies to
	\begin{align}
	\left(f^{+}\right)^{\star}(X) = & \sup_{Z \in \Lp{\ffff[]}[][p]}\ee{ZX} - f^{+}(Z)\\
	= & \sup_{Z \in \Lp{\ffff[]}[]}\ee{ZX} - f^{+}(Z)\\
	= & \sup_{Z \in \Lp{\ffff[]}[]}\ee{ZX} - f(Z),
	\label{eq_optim_simpl_1}
	\end{align}
	where the fact that $f^{+}|_{\Lp{\ffff[]}[][p]}=f$ was used in the last line.  
\end{proof}
\begin{proof}[Proof of Proposition~\ref{prop_approx_ext}]
	Since $M_0$ then $\lambda=1$.  Therefore, from Example~\eqref{ex_Lagrangian_mRC} it follows that $\rrrr_{\eta}(X)$ must solve
	\begin{equation}
	\begin{aligned}
	\rrrr_{\eta}(X) %
	= & \arginf{Z \in \Lp{\ffff[]}[d] }[1]
	\frac1{2}
	\ee{
		\left\|
		X-Z
		\right\|^2
	}
	+
	\eta
	\rho(X)\\
	=&
	\operatorname{Prox}_{\eta\rho+\iota_{\Lp{\gggg}}[d]}(X)
	.
	\end{aligned}
	\label{eq_lalaland_reparameterizos}
	\end{equation}
	From~\citep[Remark 12.24]{ConvexMonoCMS}, it follows that $\rho_{+:\eta}(X)$ is the Moreau-Yoshida envelope of $\eta\rho+\iota_{\Lp{\gggg}}[d]$ evaluated at $X$.  Therefore (i)-(iii) follow from \cite[Propositions 12.15,%
	12.30,12.33]{ConvexMonoCMS}.

	For Property (iv), applying \cite[Proposition 15.1]{ConvexMonoCMS} to the infimal convolution defined by~\eqref{eq_lalaland_reparameterizos} we obtain
	\begin{equation}
	\rho_{+:\eta}^{\star}(X) = \left(
	\frac1{2\eta}
	\ee{
		X^2    
	}\right)^{\star}
	+
	\left(\rho^+\right)^{\star}(X)
	\label{eq_lalaland_2_inf_conv_simpl_more_0}
	.
	\end{equation}
	Applying Lemma~\ref{lem_conj} to~\eqref{eq_lalaland_2_inf_conv_simpl_more_0} yields
	\begin{equation}
	\rho_{+:\eta}^{\star}(X) = \left(
	\frac1{2\eta}
	\ee{
		X^2    
	}\right)^{\star}
	+
	\rho^{\star}(X)
	\label{eq_lalaland_2_inf_conv_simpl_more}
	.
	\end{equation}
	Applying the characterization of self-conjugacy of \cite[Proposition 13.19]{ConvexMonoCMS} the scaling transformation laws for Legendre-Fenchel conjugacy to~\eqref{eq_lalaland_2_inf_conv_simpl_more} yields
	\begin{equation}
	\begin{aligned}
	\rho_{+:\eta}^{\star}(X) = &
	\left(
	\frac1{2\eta}
	\ee{
		X^2    
	}\right)^{\star}
	+
	\rho^{\star}(X)
	\\
	= &
	\frac{1}{2\eta
	}
	\ee{
		\left(\eta X\right)^2
	}
	+
	\rho^{\star}(X)
	.
	\end{aligned}
	\label{eq_lalaland_2_inf_conv_simpl_even_more}
	\end{equation}
	
	\cite[Theorem 3.1]{filipovic2007convex}, implies that since $\rho$ is convex on $\Lp{\ffff[]}[]$, then the minimal penalty of $\rho_2$ satisfies
	\begin{equation}
	\left(\rho_2\right)^{\star}=\rho^{\star}|_{\Lp{\ffff[]}[]}
	\label{eq_lelink}
	.
	\end{equation}
	Combining~\eqref{eq_lalaland_2_inf_conv_simpl_even_more} and~\eqref{eq_lelink} yields (iv).  
\end{proof}
Therefore, even though $\rho_{+:\eta}$ is not a risk-measure, it provides a regular extension of $\rho$ to $\Lp{\ffff[]}[]$ and approximates $\rho^+$ and the minimal penalty of $\rho_2$, to arbitrary precision.

Next, a numerical implementation comparing the risk-neutral price of an option to the risk-averse price is considered.  
\section{Example: Risk-Averse Vanilla Option Value}\label{s_example_RA_Vanilla}
A natural first example is to examine the behaviour of the risk-averse value, as compared to the risk-neutral value of a plain vanilla option in the classical setting of \cite{black1973pricing}.  In this setting, the price of a stock is assumed to satisfy a geometric Brownian motion,
$$
X_t = X_0\exp\left(
\left(\mu - \frac{\sigma^2}{2} \right)t + \sigma W_t
\right);\qquad X_0=x_0
$$
where the initial price $x_0$ is a positive real number, $W_t$ is an $\rr$-valued Brownian motion, the stock's drift $\mu$ is any real constant, and its volatility $\sigma$ is a positive number.  

The \cite[Black-Scholes-Merton formula]{black1973pricing}, gives an expression for the price of a European Call (resp. Put) option, whose (discounted) payoff at the maturity time $T>0$, is given by
$$
f(x)\triangleq e^{-rT}\max\{x-K,0\}
.
$$
Here $K>0$, is the predetermined strike price, set at the issue of the derivative contract, and $r\geq 0$, is the risk-free rate in effect.  In the setting of \cite{black1973pricing}, it is assumed that $r$ is constant.  
For the remainder of this example, the risk-perspective will be $\rrrr\triangleq \left(
\ffff[0],\rho_2,\utility_{\hhh},\Lp{\ffff[]},M_{\lambda}
\right)$.  

For the first implementation, set $T=10$ years, the interest rate $r=.0175$, $K=.2$, $X_0=1$, and $\sigma=.1$. The parameters used in the algorithm of Corollary~\ref{cor_simplified_algo}, will be $\gamma=.7$, the forward-backwards splitting procedure will stop after $10^4$ iterations, and in each Monte-Carlo estimation of the conditional expectation under $\qq$ will use $10^4$ particles.  

As illustrated by Figure~\ref{fig_Call_evolution_lambda}, as the value of the Vanilla option decreases as the
mispricing risk-aversion level $M_{\lambda}$, level increases.  Moreover, when $\lambda\approx0$, the risk-averse value is approximately equal to the classical risk-neutral price of the Vanilla option.  
\begin{figure}[H]%
	\centering
	\includegraphics[width=0.8\textwidth]{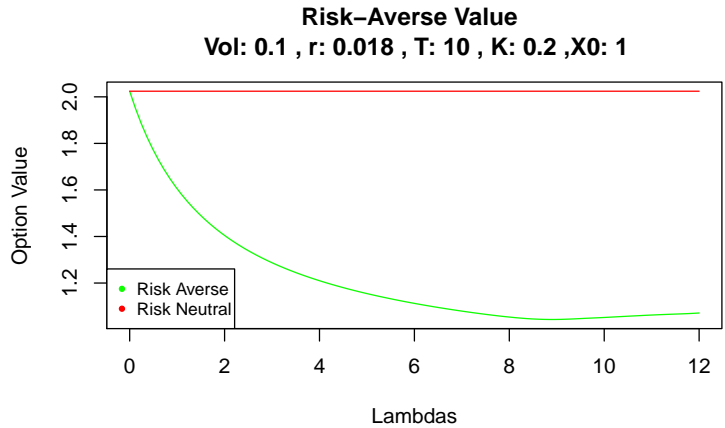}
	\label{fig_Call_evolution_lambda}
	\caption{The Lambda axis is parameterized according to $\tilde{\lambda}=\frac{2\lambda}{1-\lambda}$}
\end{figure}

Tables~\ref{Comparison_of_Vanilla_Call_5},~\ref{Comparison_of_Vanilla_Call_4}~\ref{Comparison_of_Vanilla_Call_3} %
compare
the quantities $V^{\rho_2^{\utility_{\hhh}}}$, $\rho_2\left(\rrrr({f(X_T)})\right)$, and 
$\epsilon^{f}_T(Z)$ %
will respectively.  For legibility, these will be respectively denoted by \textit{Option Value}, \textit{Estimator Risk}, and \textit{Mispricing Risk}.  

\begin{table}[H]
	\centering
	\begin{tabular}{rrrr}
		\hline
		& Risk-Averse & Risk-Neutral & RA/RN \\ 
		\hline
		Option Value & 1.993 & 2.056 & 0.970 \\ 
		Estimator Risk & -1.567 & -1.538 & 1.019 \\ 
		Mispricing Risk & 0.800 & 0.846 & 0.946 \\ 
		\hline
	\end{tabular}
	\caption{Lambda= 0.1} 
	\label{Comparison_of_Vanilla_Call_5}
\end{table}
\begin{table}[H]
	\centering
	\begin{tabular}{rrrr}
		\hline
		& Risk-Averse & Risk-Neutral & RA/RN \\ 
		\hline
		Option Value & 1.767 & 2.028 & 0.871 \\ 
		Estimator Risk & -1.512 & -1.478 & 1.023 \\ 
		Mispricing Risk & 0.681 & 0.870 & 0.783 \\ 
		\hline
	\end{tabular}
	\caption{Lambda= 0.5} 
	\label{Comparison_of_Vanilla_Call_4}
\end{table}
\begin{table}[H]
	\centering
	\begin{tabular}{rrrr}
		\hline
		& Risk-Averse & Risk-Neutral & RA/RN \\ 
		\hline
		Option Value & 1.619 & 2.046 & 0.791 \\ 
		Estimator Risk & -1.604 & -1.586 & 1.012 \\ 
		Mispricing Risk & 0.508 & 0.802 & 0.633 \\ 
		\hline
	\end{tabular}
	\caption{Lambda= 1} 
	\label{Comparison_of_Vanilla_Call_3}
\end{table}
The preceding tables show that as $\tilde{\lambda}\triangleq \frac{\lambda}{1-\lambda}$ increases, the risk-averse value of the option decreases.  This reflect the investors' disinterest in purchasing an option if there is a high amount of mispricing risk.  As the investor's risk-aversion level increases, then so does their perceived value of the option.  Here, the relationship between $\lambda$ and the mispricing-risk aversion level, described by~\eqref{key}, is implicitly exploited.  

We observe that a reduction in the risk of the mispricing-risk does not imply a reduction in the risk of the estimator $\rho_2\left(\rrrr({f(X_T)})\right)$; actually, the opposite trend is seen.  Namely, that a reduction in the mispricing-risk, increases the inherent riskiness of $V^{\rrrr}$ itself.

\section{Conclusion}\label{s_conclusion}
This paper introduced a non-linear extension of the conditional expectation operator on $\Lp{\ffff[]}$, which served as an alternative to sublinear expectation, called the $\rrrr$-conditioning operator.  
Theorems~\ref{thrm_existsnce_uniqueness} and~\ref{thrm_Lagrangian_Formulation}, proved that the $\rrrr$-conditioning operator is well-defined %
and solves~\eqref{PG}.  

Theorem~\ref{thrm_risk_reduction} showed that $\rrrr$-conditioning provides a middle ground between the classical risk-neutral pricing problem~\eqref{RN} and the classical robust finance extension of it.  This middle ground is interpreted as the investor having partial uncertainty towards their model, instead of full certainty or full uncertainty.  Moreover, in Theorem~\ref{thrm_Approximation_Theorem} it was proven that as the parameter $\lambda$, quantifying the degree of model uncertainty, approaches its maximum $\rrrr(X)$ converges to a (penalized) conditional sublinear expectation of $X$, in the $L^2$-sense.  In this way, by appealing to the computational tools available for computing $\rrrr(X)$, we may approximate conditional sublinear expectations to arbitrary precision by turning to $\rrrr(X)$.  

The $\rrrr$-conditioning operators, introduced in this paper, provide a natural way to incorporate various levels of risk-aversion into a derivative security's pricing scheme, under the assumption of market completeness.  Theorem~\ref{thrm_compute} provided a strongly convergent, forward-backward splitting algorithm for computing this price, as well as the $\rrrr$-conditioning in general.  The algorithm was implemented and used to price a plain Vanilla Call in the Black-Scholes-Merton framework.  We found that the risk-averse price of the option decreases as the risk-aversion level increased; which is consistent with our economic intuitions.  

We explored examples of $\rrrr$-conditioning operators which were well-suited to high-dimensional data.  In Proposition~\ref{prop_sparse_cond_exp}, it was proven that these operators have a high probability of providing sparser estimates than their classical conditional expectation counterparts.  We called this \textit{sparse conditional expectations}, and they are computed using the strongly convergent algorithm provided by Proposition~\ref{prop_computable}.  

In further applications, we showed that $\rrrr$-conditioning could be used to provide smooth approximations to convex risk-measures, which was interpreted through the lens of robust finance in two ways.  The first of these two interpretations was a robust representation for the optimal value of $\riskutility\left(\eeRA{X}\right)$, given in Proposition~\ref{prop_rob}, and the second, is shown in Proposition~\ref{prop_rob_neutral}, was the neutrality of $\eeRA{X}$ for a robust sensitivity.  

We believe that the regularity, connections to convex analysis, approximation capabilities, and the scope of $\rrrr$-conditioning make it a refreshing new object of study in mathematical finance.  Furthermore, we expect that it applies to any problem where partial uncertainty arises, as well as other applications not considered in this paper.  

\bibliographystyle{abbrvnat}
\bibliography{References}

\end{document}